\documentclass[preprint,5p,times,twocolumn,authoryear]{elsarticle}

\usepackage{lineno,hyperref}

\usepackage[cmex10]{amsmath}
\usepackage{amssymb}
\usepackage{mathrsfs}
\usepackage{amsthm}
\usepackage{graphicx}
\usepackage{color}

\usepackage{bm}

\DeclareGraphicsExtensions{.pdf,.png,.jpg,.eps,.ps}

\def\ba{\begin{array}}
	\def\ea{\end{array}}
\def\baa{\begin{align}}
	\def\eaa{\end{align}}

\newcommand{\bsq}{\begin{subequations}}
	\newcommand{\esq}{\end{subequations}}

\newcommand{\beq}{\begin{equation}}
\newcommand{\eeq}{\end{equation}}
\newcommand{\bq}{\begin{eqnarray}}
\newcommand{\eq}{\end{eqnarray}}
\newcommand{\bqn}{\begin{eqnarray*}}
	\newcommand{\eqn}{\end{eqnarray*}}
\newcommand{\bee}{\begin{enumerate}}
	\newcommand{\eee}{\end{enumerate}}
\newcommand{\bi}{\begin{itemize}}
	\newcommand{\ei}{\end{itemize}}

\usepackage{comment}
\usepackage{ifthen}
\newboolean{showcomments}
\setboolean{showcomments}{false}
\newcommand{\wang}[1]{\ifthenelse{\boolean{showcomments}}
	{ \textcolor[rgb]{1,0,1}{(ZW:  #1)}}{}}
\newcommand{\fliu}[1]{\ifthenelse{\boolean{showcomments}}
	{ \textcolor{red}{(FL:  #1)}}{}}
\newcommand{\slow}[1]{\ifthenelse{\boolean{showcomments}}
	{ \textcolor{blue}{(SL:  #1)}}{}}

\theoremstyle{definition}
\newtheorem{theorem}{Theorem}
\newtheorem{lemma}[theorem]{Lemma}
\newtheorem{corollary}[theorem]{Corollary}

\theoremstyle{definition}
\newtheorem{definition}{Definition}
\newtheorem{remark}{Remark}

\allowdisplaybreaks

\journal{xxx}

\begin{document}

\begin{frontmatter}

\title{Distributed Load-Side Control: Coping with  Variation of Renewable Generations}

\tnotetext[mytitlenote]{This work was supported  by the National Natural Science Foundation of China ( No. 51677100, U1766206, No. 51621065), the US National Science Foundation through awards EPCN 1619352, CCF 1637598, CNS 1545096, ARPA-E award DE-AR0000699, and Skoltech through Collaboration Agreement 1075-MRA.}

\author[thu]{Zhaojian~Wang}
\author[thu]{Shengwei~Mei}
\author[thu]{Feng~Liu\corref{mycorrespondingauthor}}\ead{lfeng@mail.tsinghua.edu.cn}
\author[caltech]{Steven H. Low}
\author[thu]{Peng~Yang}

\cortext[mycorrespondingauthor]{Corresponding author}

\address[thu]{State Key Laboratory of Power Systems, Department of Electrical Engineering, Tsinghua University, Beijing 100084, China}
\address[caltech]{Department of Electrical Engineering, California Institute of Technology, Pasadena, CA 91105, USA}

\begin{abstract}
	{\color{black}This paper addresses the distributed frequency control problem in a multi-area power system taking into account of unknown time-varying  power imbalance. Particularly, fast controllable loads are utilized to restore system frequency under changing power imbalance in an optimal manner.} The imbalanced power  causing frequency deviation is decomposed into three parts: a known constant part, an unknown low-frequency variation and a high-frequency residual. The known steady part is usually the prediction of power imbalance. The variation may result from the fluctuation of renewable resources, electric vehicle charging, etc., which is usually unknown to operators. The high-frequency residual is usually unknown and treated as an external disturbance. Correspondingly, in this paper, we resolve the following three problems in different timescales: 1) allocate the steady part of power imbalance economically; 2) mitigate the effect of unknown low-frequency power variation locally; 3) attenuate unknown high-frequency disturbances. To this end, a distributed controller combining consensus method with  adaptive internal model control is proposed. We first prove that the closed-loop system is asymptotically stable and converges to the optimal solution of an optimization problem if the external disturbance is not included. We then prove that the power variation can be  mitigated accurately. Furthermore, we show that the closed-loop system is robust against both  parameter uncertainty and external disturbances. The New England system is used to verify the efficacy of our design. 
\end{abstract}

\end{frontmatter}

\section{Introduction}
\subsection{Background}
In the modern power system, multiple regional grids are usually interconnected to constitute a bulk grid \cite{Min:Total, Ahmadi:Probabilistic}. To maintain a stable power system, the frequency should be retained at its nominal value, e.g. 50Hz or 60Hz. Conventionally, it is realized by synchronized generators in a centralized fashion, known as a hierarchy control architecture \cite{Kundur:Power,dorfler2016breaking}. However, with the increasing penetration of volatile and uncertain renewable generations, power mismatch in the system can fluctuate rapidly with a large amount. In such a situation, the traditional manner of control may not be able to keep pace due to large inertia of the traditional synchronous generators.
Fortunately, load-side participation in frequency control opens up new possibility to resolve this problem, benefiting from its fast response \cite{Schweppe1980Homeostatic, Changhong:Design}. 
On the other hand, as controllable loads are usually dispersed across the power system, a distributed  architecture is more suitable for load-side control than the conventional centralized one. Indeed, distributed optimal control has been investigated by combining controller design with optimal dispatch problems \cite{jokic:real,zhang:real, Stegink:aunifying}. It leads to a so-called \emph{reverse engineering} methodology for designing optimal  controllers, particularly in optimal frequency control of power systems \cite{Li:Connecting, Changhong:Design, Cai:Distributed}. In this paper, we  design a distributed  load-side controller that is capable of adapting to power variation due to volatile renewable generations, such as wind farms and PV clusters.  
\subsection{Related Work}
In power system operation, frequency deviation is usually a consequence of power mismatch due to unexpected disturbances, such as sudden load leaping/dropping or generator tripping. Frequency control papers can be roughly divided into two categories in terms of the forms of power imbalance: constant power imbalance \cite{Changhong:Design, Mallada-2017-OLC-TAC, wang:nonlinear, Kasis:Primary1, Distributed_I:Wang, Distributed_II:Wang, Lu2016Distributed} and time-varying power imbalance \cite{trip2016internal,Xi2017Power,weitenberg2017robust}. In the first category,  a step change of load/generation is considered. Then generators and/or controllable loads are utilized to eliminate the power imbalance and restore the nominal frequency. In \cite{Changhong:Design}, an optimal load-side control problem is formulated and a primary frequency controller is derived to balance step power change using controllable loads. It is extended in \cite{Mallada-2017-OLC-TAC} to realize a secondary frequency control, i.e. restoring the nominal frequency. The design approach is  generalized in \cite{Kasis:Primary1}, where the  model requirement is relaxed and a passivity condition is proposed to guarantee asymptotic stability. \cite{Distributed_I:Wang,Distributed_II:Wang} further consider both steady-state and transient operational constraints in distributed optimal frequency control. In \cite{wang:nonlinear},  a nonlinear network-preserving model is considered and only limited control coverage is needed to implement the distributed optimal frequency control. A different disturbance  is considered in \cite{Lu2016Distributed}, where the secondary frequency controller is injected by constant malicious attacks. To eliminate the influence of the attacks, a detection method is derived to combine with the distributed frequency controller. 
%{\color{black}In this category, the imbalanced power is regarded a constant}

In the second category, power imbalance is not constant, creating much greater challenge to controller design and stability analysis. In \cite{trip2016internal},  power variation is modeled as  output of a known exosystem. Then an internal model controller is designed to tackle and compensate for the time-varying imbalanced power. The idea of combining  distributed control with internal model control is attractive and  inspiring. In \cite{Xi2017Power}, a centralized controller is proposed, which can track the power imbalance and maintain the system frequency within a desired range in the presence slowly changing power imbalance. The frequency still varies along with time-varying loads. In \cite{weitenberg2017robust}, measurement noise is considered in frequency control, and a leaky integral controller is proposed that can strike an acceptable trade-off between performance and robustness. 

To sum up, in most of the existing literature, power disturbance is modeled as a step change. The time-varying power disturbance is usually regarded as output of a \emph{known} exosystem. However, neither model is realistic for practical power systems, especially when a large amount of renewable generations and electric vehicles are integrated. In such a situation,  power imbalance is always time-varying and unknown, which should be carefully considered in the design of distributed frequency control.

\subsection{Contribution}
In this paper, power imbalance is decomposed into a known constant part, an unknown low-frequency time-varying part and a high-frequency residual. In power systems, the first one can be obtained by prediction while the latter two are fluctuations around the prediction. Offset error in prediction can also be considered in the unknown time-varying part. {\color{black} This decomposition suggests a way to deal with time-varying disturbances. First, a distributed control is proposed based on consensus method to balance the known constant part economically, which resolves a slow timescale operation problem. Second, a decentralized supplementary controller based on the internal model control is proposed to mitigate the effect of the unknown low-frequency variation at a faster timescale. Finally, we also ensure that the proposed controller attenuate the impact of high-frequency residual.} 

\textcolor{black}{This work can be regarded as an extension of 
	\cite{Changhong:Design, Mallada-2017-OLC-TAC, wang:nonlinear, Kasis:Primary1, Distributed_I:Wang, Distributed_II:Wang}. As the power imbalance is time varying in our case,  these previous distributed  controller may not be able to stabilize and restore the frequency, as we will demonstrate later in case studies. Here the main challenge is how to fit a time-varying tracking and compensation control into the structure of the previous distributed frequency controller.} The major difference between this paper and \cite{trip2016internal} is that the power variation  is modeled as output of a known exosystem in \cite{trip2016internal}. Since such information is difficult to obtain in practice, our model appears to be more practical. In \cite{Xinghu:Distributed}, an internal model control is leveraged to devise a distributed \emph{unconstrained} optimization which can mitigate the effects of unknown time-varying disturbances. In contrast, we consider optimal frequency control problem with both \emph{power system  dynamics} as well as \emph{power balance constraints}, which are not included in \cite{Xinghu:Distributed}. Moreover, we also analyze the robustness of the proposed controller under uncertain parameters and disturbances. Main contributions of this paper are as follows: 
\begin{itemize}
	\item A generic model of power imbalance for frequency control is established, consisting of three parts: a known constant part, an unknown low-frequency power variation and a high-frequency residual. The power variation is further modeled by a superposition of several dominant sinusoidal components. Then it is formulated as the output of an exosystem with unknown parameters; 
	\item A distributed controller is derived to restore the nominal frequency even under unknown disturbance. It is composed of two parts. One is designed based on consensus control to achieve an economic allocation of the constant part of power imbalance, while the other is designed based on adaptive internal model control to mitigate the effect of unknown power variation;
	\item Robustness of the controller under parameter uncertainty and external disturbances is analyzed. It is proved that the uncertain damping constant has no impact on the performance of the controller and the impact of external disturbances is attenuated greatly. 
\end{itemize}
\subsection{Organization}	
The rest of this paper is organized as follows. In Section 2, the network and power imbalance models are formulated. Section 3 presents the design of distributed frequency controller. In Section 4, the equilibrium of the closed-loop system is characterized with a proof of asymptotic stability. The robustness of the proposed controller under uncertainties is analyzed in Section 5. We confirm controller performance via simulations in Section 6. Section 7 concludes the paper.

\section{Problem Formulation}

\subsection{Model of Power Network }

A large power network is usually composed of multiple control areas, which are interconnected
through tie lines. For simplicity, we treat each control area
as a node with an aggregate controllable load and an aggregate uncontrollable power injection.\footnote{In our study, all controllable loads in the same area are aggregated into one controllable load. The same for the aggregate uncontrollable power injection. This simplification is practically reasonable when dealing with the frequency control problem in power systems \cite{Li:Connecting}.}
Then the power network is modeled as a graph ${\cal G }:=(N, E)$ 
where  $N=\{1,2,...n\}$ is the set of nodes (control areas) and
$E\subseteq N\times N$ is the set of edges (tie lines). If a pair of
nodes $i$ and $j$ are connected by a tie line
directly, we denote the tie line by $(i,j)\in E$. $\cal G$ is treated as directed with an arbitrary orientation and we use $(i,j)\in E$ or $i\rightarrow j$ interchangeably to denote a directed edge from $i$ to $j$. Without loss of generality, we assume ${\cal G }$ is connected.

Besides the  graph of physical power network, we also need to consider the communication network, modeled by a graph ${\cal H }$ whose nodes are the same set $N$ of graph ${\cal G }$ with
possibly a different set of edges.
An edge in ${\cal H }$ means that the two endpoints of the edge can communicate with each other directly. In this paper, we assume  ${\cal H }$ is also connected.  The set of neighbors of node $j$ in the communication graph ${\cal H }$ is denoted by $N_{cj}$. The Laplacian matrix of ${\cal H }$ is denoted as $L$.

A second-order linearized model is adopted to describe the frequency dynamics of each node. We assume the tie lines are lossless and adopt the DC power flow model, which is reasonable for a high-voltage transmission system. Then for each node $j\in N$, we have
\begin{subequations}
	\label{system model}
	\begin{align}
	\dot \theta_j\ & = \ \omega_j
	\label{eq:model.1a}
	\\
	M_j \dot \omega_j\ & =\   P^{in}_j - P^{l}_j  -D_j \omega_j \nonumber\\
	&\quad + \sum\nolimits_{i: i\rightarrow j} \! B_{ij} (\theta_i - \theta_j)
	-  \sum\nolimits_{k: j\rightarrow k} \! B_{jk}(\theta_j - \theta_k)
	\label{eq:model.1b}
	\end{align}
\end{subequations}
where, $\theta_j$ denotes the rotor angle at node $j$; $\omega_j$  the frequency deviation; $P^{in}_j$ the uncontrollable power injection; $P^{l}_j$ the controllable load. $M_j>0,\ D_j>0$ are inertia and damping constants, respectively. $B_{jk}>0$ are line parameters that depend on the reactances of  line $(j,k) \in E$.

\subsection{Model of Power Imbalance}
Denote $P^{in}_j$ as the imbalanced power in the system. It can be decomposed into two parts: a constant part and a variation part. That is
\begin{align}
\label{power disturbance}
P_j^{in}(t)=\overline P_j^{in}+\tilde q_j(t)
\end{align}
where $\overline P_j^{in}$ is the known constant part, which  could be the prediction of renewable generations and/or loads. $\tilde q_j(t)$ is the variation part, which is assumed unknown.\footnote{As $\overline P_j^{in}$ may not be accurate, the offset error of prediction is included in $\tilde q_j(t)$ component. We abuse the term $\tilde q_j(t)$ ``variation part" for simplicity. } 

The known constant part is easy to deal with, while the variation part is non-trivial. The main idea is to further decompose it into the sum of a series of sinusoidal functions, whose parameters are unknown. Then an internal model control can be utilized to trace these sinusoidal components, and then  eliminate the effects of the variation part. 

In light of \cite{Milan2013Turbulent, buvsia2016distributed, barooah2015spectral, AGUIRRE2008Dynamical}, we can approximate variation of renewable generations and load demands by a superposition of a few sinusoidal functions. 
%In addition, low frequency oscillation is inherently superposition of sinusoidal signals \cite{ Bauer1975Simulation}. 
Specifically, we decompose the power imbalance $\tilde q_j(t)$ at node $j$ injected by volatile renewable generation and loads   into 
\begin{align}
	\label{sin_functions}
	\tilde q_j(t):=q_{j0}+\sum\nolimits_{k=1}^{s_j}q_{jk}\sin(a_{jk}\cdot t+\phi _{jk}) + w_j(t)
\end{align}
where $q_{j0}$ is the prediction offset error (which is an unknown constant). The second term models the variation part, which is a superposition of $s_j$ sinusoidal functions. Their amplitudes  $q_{jk}$, frequencies $a_{jk}> 0$ and initial phases $\phi _{jk}$ are unknown but belong to a known bounded interval. 
Here we consider only a few low-frequency power fluctuations. The remaining  high-frequency residuals, denoted by $w_j(t)$, is usually quite small. So we treat $w_j(t)$ as an external disturbance and do not consider its detailed model in this paper, but simply assume that it belongs to the $\mathcal{L}_2^T$ space, i.e., for any $w_j(t)\;\;(j\in N)$,  
$\int_{0}^{T}||w_j(t)||^2\text{d}t < +\infty$ holds for all $0<T<+\infty$.

\begin{remark}[Power Variation in power system]
	In this paper, we adopt a generic model to depict $\tilde q_j(t)$ so that it is applicable to various types of power imbalance in practice. In practical power systems, $\tilde q_j(t)$ has many interpretations, some of which are listed below.
	\begin{enumerate}[1)]
		\item Variation of renewable generations. Large-scale renewable  generations may vary rapidly. As it is difficult to accurately predict volatile renewable generations,  the fluctuation is always partly unknown. Such unknown variations may lead to severe frequency fluctuations or even instability.
		\item Variation of loads. Load demands in a power system are always varying. Whereas load demand usually can be estimated quite accurately in a traditional power system,  the integration of electric vehicles, energy storage and  demand response makes demands much more difficult to predict. 
%		\item Low frequency oscillation. This is a critical problem in power system stability, which is caused by many reasons. Particularly, a local weak damping mode (e.g. inappropriate parameters in a generator's governor) in a power system may 
%		excite global power and frequency oscillations. However, as the  governor may not be accurately modeled, the oscillation modes are also unknown. 
	\end{enumerate}

{\color{black}	We use a generic form to represent the variation of renewable generations and loads instead of detailed models of wind generators and PVs. Actually, it is common to treat power variation due to wind generators, PVs and loads as an aggregated injection \cite{trip2016internal,Li:Connecting,Mallada-2017-OLC-TAC}. Here we follow this treatment. }
\end{remark}

\subsection{Equivalent Transformation of Disturbance Model}
%To investigate the variation part, 
We further investigate the dominant part in $\tilde q_j(t)$. Denote
\begin{align}
\label{variation}
q_j(t):=q_{j0}+\sum\nolimits_{k=1}^{s_j}q_{jk}\sin(a_{jk}\cdot t+\phi _{jk})
\end{align}
%As the residuals $w_j(t)$ usually have small amplitudes, 

Then we  show that $q_j(t)$ can be expressed as the output of an exosystem. To this end, define 
\begin{align}
\begin{split}
	\lambda_{j1}&\ =\ q_j(t)\\
	\lambda_{jk}&\ =\ \left(\frac{\text{d}}{\text{d}t}\right)^{k-1} q_j(t) \quad (2\le k\le \bar s_j)
\end{split}
\end{align}
where  $\bar s_j:=2s_j+1$. Then $q_j(t)$ is just the output of the following dynamic system \cite{Pulido:A,Xinghu:Distributed}:
\bsq
\begin{align}
\label{lambda1}
\dot\lambda_{j}&=A_j(\alpha_j)\lambda_{j}\\
q_j(t)&=[1\ \textbf{0}_{1\times 2s_j}]\cdot \lambda_{j}
\end{align}
\esq
where,
\begin{align}
\begin{split}
\lambda_{j}\ &:=\ [\lambda_{j1}, \cdots, \lambda_{j,\bar s_j}]^T \\ 
A_j(\alpha_j)\ &:=\  \left [ \begin{array}{*{20}{c}}
	\textbf{0}_{2s_j\times 1} & I_{2s_j}\\
	0  & \alpha_{j1}, 0, \cdots, \alpha_{js_j}, 0
\end{array}  \right ]
\label{matrix:A}
\end{split}
\end{align}
with 
$\alpha_{j1}=-\prod\limits_{l=1}^{s_j}a_{jl}^2 $, $\alpha_{j2}=-\sum\limits_{k=1}^{s_j}\prod\limits_{\scriptstyle\hfill l=1\atop\scriptstyle\hfill l\neq k}^{s_j}a_{jl}^2$, $\cdots$, $\alpha_{js_j}=-\sum\limits_{l=1}^{s_j}a_{jl}^2$. Here, $a_{jl}$ are defined in \eqref{variation}.

To facilitate the controller design, a transformation is constructed. Let $R_j:=[r_{i1}, \cdots, r_{i,\bar s_j-1}, 1]$, such that all the roots of polynomial $\tau^{\bar s_j-1}+ r_{i,\bar s_j-1}\tau^{\bar s_j-2}+\cdots+r_{i2}\tau+r_{i1}$ have negative real parts. Then define a vector $\tilde A_j(\alpha_j):=R_j(I_{\bar s_j}+A_j(\alpha_j))$ and construct the following matrix 
\begin{align}
O_j(\alpha_j):=\left[\tilde A_j^T(\alpha_j), \cdots, \left(\tilde A_j(\alpha_j)A_j^{\bar s_j-1}(\alpha_j)\right) ^T\right]^T\nonumber
\end{align}
In \cite{XU2016Output} , it is proven that $ O_j(\alpha_j)$ is nonsingular, and 
\begin{align*}
O_j^{-1}(\alpha_j)A_j(\alpha_j)O_j(\alpha_j)=A_j(\alpha_j)
\end{align*}

Let $\varphi_j=O^{-1}_j(\alpha_j)\lambda_j$. Then we have
\bsq
\begin{align}
	\label{varphi1}
	\dot\varphi_j&=A_j(\alpha_j)\varphi_j\\
	\label{varphi2}
	q_j(t)&=\tilde A_j(\alpha_j)\varphi_j
\end{align}
\esq
So far,  $q_j(t)$ is written as the output of a new exosystem \eqref{varphi1}. However, elements in $A_j(\alpha_j)\ \text{and}\ \tilde A_j(\alpha_j)$ are still unknown. According to the definition of $q_j(t)$ and boundedness of $q_{jk}, a_{jk}$, we have $\lambda_{j}$ is  bounded. Hence $\varphi_j$ is also bounded due to the nonsingular transformation. 
%\fliu{I think the description is a little bit confusing. }

From \eqref{power disturbance}, \eqref{sin_functions} and \eqref{variation}, $P_j^{in}(t)$ is composed of three parts, i.e., $\overline P_j^{in}$, $q_j(t)$ and $w_j(t)$, we will address them in different ways, giving rise to the following three problems.
\begin{itemize}
	\item[\textbf{P1}:] Balancing $\overline P_j^{in}$ economically and globally;
	\item[\textbf{P2}:] Coping with the variation of $q_j(t)$ locally;
	\item[\textbf{P3}:] Attenuating the impact of external disturbance $w_j(t)$.
\end{itemize}

{\color{black}\begin{remark}[Timescales]
		The above three problems can be interpreted from the perspective of multiple timescales in power systems. \textbf{P1} is the long-term operation problem, i.e., the system should operate economically in a steady state, where the time scale is about several minutes. \textbf{P2} is the short-term control problem with time scale of several seconds, where the low-frequency variation should be eliminated by designing proper controller. The timescale of \textbf{P3} is even faster than that of \textbf{P2}, where the controller cannot track the high-frequency disturbance accurately. In this situation, we hope to attenuate its negative impact. Thus, we resolve the distributed frequency control problem under time-varying power imbalance systematically in three different timescales, which coincides with  \textbf{P1-P3}.
\end{remark}}

\section{Controller Design}
In this section, the known steady-state part $\overline P_j^{in}$ is optimally balanced across all areas using a consensus-based distributed control, which resolves \textbf{P1}. Then the effect of variation part $q_j(t)$  is eliminated locally by using a supplementary internal model controller, resolving \textbf{P2}.  In terms of \textbf{P3}, here we do not design a specific controller to deal with $w_j(t)$. Instead we show that the proposed controller can effectively attenuate  $w_j(t)$, which will be discussed in Section V. 
%Then, a complete controller is proposed. 

\subsection{Controller for the Known Steady-state $\bar{P}_j^{in}$}
First we  formulate an optimization model for the optimal load control problem:
\begin{subequations}
	\label{eq:opt.1}  
	\begin{align}
		\text{OLC:~}\  \min\limits_{P_j^l} &\quad   \sum\nolimits_{j\in N} \frac{1}{2}\beta_j \cdot \left(P^{l}_j\right)^2
		\label{eq:opt.1d}
		\\ 
		\text{s. t.}  
		&  \quad	\sum\nolimits_{j\in N} \overline P^{in}_j=\sum\nolimits_{j\in N} P^{l}_j
		\label{eq:opt.1a}
	\end{align}   
\end{subequations}
where $\beta_j>0$ are constants.
The control goal of each area is to minimize the regulation cost of the controllable load, which is in a quadratic form \cite{trip2016internal}.  \eqref{eq:opt.1a} is the power balance  constraint.
Suppose for that $\tilde q_j(t)=0$. We design a consensus-based controller \cite{trip2016internal}
\begin{subequations}
	\label{controller_constant}
\begin{align}
	\label{constant_2}
	P^l_j&=\mu_j/\beta_j \\
	\label{constant_1}
	\dot\mu_j&=-\sum\nolimits_{k\in N_{cj}}(\mu_j-\mu_k) + \omega_j/\beta_j 
\end{align}
\end{subequations}
{\color{black}In \eqref{constant_2}, $\mu_j$ are the consensus variables, and $-\mu_j$ stands for the marginal costs of individual controllable loads. In the steady state, all $\mu_j$ should converge to an identical value for all controllable loads when $\omega_j$ converges to zero.}

This simple controller can restore the frequency and minimize the regulation cost of the controllable loads when $\tilde{q}_j(t)=0$. However, a time-varying $\tilde{q}_j(t)$ may destroy the controller. Next we use a supplementary controller to deal with $\tilde{q}_j(t)$.

\subsection{Controller Considering Varying Power Imbalance}
In this subsection, an adaptive internal model control is supplemented to mitigate $q_j(t)$, which is given by
\begin{subequations}
	\label{controller_varying}
	\begin{align}
	\label{varying_5}
	P^l_j\ &=\ \mu_j/\beta_j +[ d_j\omega_j +\tilde A_j(\hat\alpha_j){\zeta}_j]\\	
	\label{varying_1}
	\dot\mu_j\ &=\ -\sum\nolimits_{k\in N_{cj}}(\mu_j-\mu_k) + \omega_j/\beta_j \\
	\label{varying_2}
	\begin{split}
		\dot\eta_j\ &=\ -\eta_j + \overline P^{in}_j - P^{l}_j  -D_j \omega_j
	\\
	& \ \ \ \ \  + \sum\nolimits_{i: i\rightarrow j}  B_{ij} (\theta_i - \theta_j)
	-  \sum\nolimits_{k: j\rightarrow k} B_{jk}(\theta_j - \theta_k)
	\end{split}
    \\
	\label{varying_3}
	\dot{\zeta}_j\ &=\  A_j(\hat\alpha_j){\zeta}_j-G_j(\eta_j+R_j\zeta_j)    \\
	\label{varying_4}
	\dot{\hat\alpha}_j\ &=\  -k_\alpha \Lambda_j({\zeta}_j)(\eta_j+R_j\zeta_j)
	\end{align}
\end{subequations}
where $k_\alpha>0, \gamma>0$ are constant coefficients, and 
\begin{align}
G_j&=[\ \textbf{0}_{1\times (\bar{s}_j-2)},\ 1,\ \gamma\ ]^T,\nonumber\\ 
\Lambda_j({\zeta}_j)&=[\ {\zeta}_{j2},\  {\zeta}_{j4},\ \cdots,\  {\zeta}_{j,\bar{s}_j-1}\ ]^T. \nonumber
\end{align}
Here, \eqref{varying_1} is the same as \eqref{constant_1}, which is used to synchronize  $\mu_j$ and restore frequency. Dynamics of $\eta_j, \zeta_j, \hat{\alpha}_j$ are derived from the adaptive internal model. {\color{black}Comparing \eqref{varying_2} and \eqref{eq:model.1b}, we have $\dot\eta_j=-\eta_j+M_j\dot\omega_j-\tilde q_j(t)$, which implies that $\eta_j$ is intended to estimate unknown $\tilde q_j(t)$. $\zeta_j$  reproduces the dynamics of $\varphi_j$ in \eqref{varphi1}. $\hat \alpha_j$  is the estimation of $\alpha_j$. It should be noted that $\tilde A_j(\hat\alpha_j){\zeta}_j$ in \eqref{varying_5} are the estimated values of $\tilde A_j(\alpha_j)\varphi_j$, i.e. $q_j$, in \eqref{varphi1}. It will be introduced in Section 4, $\hat\alpha_j=\alpha_j$ and ${\zeta}_j=\varphi_j$ in the steady state, leading to $q_j=\tilde A_j(\hat\alpha_j){\zeta}_j$.}

In the controller \eqref{varying_5}, $\mu_j/\beta_j$  allocates $\overline{P}_j^{in}$ economically; $\tilde A_j(\hat\alpha_j){\zeta}_j$ is the output of the internal model, which is used to eliminate $q_j(t)$ asymptotically; and $d_j\omega_j$ is used to guarantee stability and enhance robustness of the controller. 
It is illustrated in Section VI that  a low-order  internal model control suffices to trace and compensate for the power variation well. 

\subsection{Closed-loop Dynamics}
Combining \eqref{system model} with \eqref{controller_varying} and omitting $w_j(t)$, we obtain a closed-loop system.
%\begin{subequations}
%	\label{closed-loop1}
%	\bq
%		\dot{\tilde \theta}_{ij}  & =&  \omega_i- \omega_j
%	\label{closed_1}	\\
%	\dot \omega_j & =& \frac{1}{M_j}\bigg( \overline P_j^{in}+q_j(t) - \left(\mu_j/\beta_j + d_j\omega_j + \tilde A_j(\hat\alpha_j){\zeta}_j\right)  \nonumber\\
%	&& \qquad -D_j \omega_j + \sum_{i: i\rightarrow j} \! B_{ij} {\tilde \theta}_{ij}	-  \sum_{k: j\rightarrow k} \! B_{jk}{\tilde \theta}_{jk}\bigg)\\
%	\label{closed_2}
%	\dot\mu_j&=&-\sum\nolimits_{k\in N_{cj}}(\mu_j-\mu_k) + \omega_j/\beta_j \\
%	\label{closed_3}
%	\dot\eta_j&=&-\eta_j +  \overline P_j^{in} - \left(\mu_j/\beta_j + d_j\omega_j +\tilde A_j(\hat\alpha_j){\zeta}_j\right)  \nonumber\\
%	&& \qquad -D_j \omega_j + \sum_{i: i\rightarrow j} \! B_{ij} {\tilde \theta}_{ij}-  \sum_{k: j\rightarrow k} \! B_{jk}{\tilde \theta}_{jk} \\
%	\label{closed_4}
%	\dot{\zeta}_j&=& A_j(\hat\alpha_j){\zeta}_j-G_j(\eta_j+R_j\zeta_j)    \\
%	\label{closed_5}
%	\dot{\hat\alpha}_j&=& -k_\alpha \Lambda_j({\zeta}_j)(\eta_j+R_j\zeta_j)
%	\eq
%\end{subequations}
Since we are only interested in  angle difference between two areas, use ${\tilde \theta}_{ij}:=\theta_i-\theta_j$ as the new state variable. 
Then perform the following transformation
\begin{align}
\label{transformation1}
\tilde{\eta}_j:=R_j\varphi_j+\eta_j,\ \ 
\tilde{\zeta}_j:={\zeta}_j-\varphi_j,\ \  
\tilde\alpha_j:=\hat\alpha_j-\alpha_j  
\end{align}
Their derivatives are 
\begin{subequations}
	\begin{align}
	\dot{\tilde{\eta}}_j    &= R_j\dot\varphi_j+\dot\eta_j \nonumber \\
	&= R_jA_j(\alpha_j)\varphi_j-\eta_j 
	+\overline P_j^{in} + \sum\nolimits_{i: i\rightarrow j} \! B_{ij} {\tilde \theta}_{ij}
	\nonumber\\
	&\quad -  \sum\nolimits_{k: j\rightarrow k} \! B_{jk}{\tilde \theta}_{jk}  -D_j \omega_j  - \left(\mu_j/\beta_j + d_j\omega_j +\tilde A_j(\hat\alpha_j){\zeta}_j\right)   \nonumber \\
	&= -{\tilde{\eta}}_j +\tilde A_j(\alpha_j){\varphi}_j -\tilde A_j(\hat\alpha_j){\zeta}_j +\overline P_j^{in} + \sum\nolimits_{i: i\rightarrow j} \! B_{ij} {\tilde \theta}_{ij} 
	\nonumber\\
	&\quad - \sum\nolimits_{k: j\rightarrow k} \! B_{jk}{\tilde \theta}_{jk}  -D_j \omega_j  - \mu_j/\beta_j - d_j\omega_j  \\                      
	\dot{\tilde{\zeta}}_j   &= \dot{\zeta}_j-\dot\varphi_j = A_j(\hat\alpha_j){\zeta}_j-G_j(\eta_j+R_j\zeta_j) - 
	A_j(\alpha_j)\varphi_j   \nonumber\\
	&= (A_j(\hat\alpha_j)-A_j(\alpha_j)){\zeta}_j + A_j(\alpha_j)({\zeta}_j-\varphi_j) \nonumber\\
	&\quad - G_j(\eta_j+R_j\zeta_j-R_j\varphi_j+R_j\varphi_j)\nonumber \\
	&= (A_j(\hat\alpha_j)-A_j(\alpha_j)){\zeta}_j + (A_j(\alpha_j)-G_jR_j){\tilde{\zeta}}_j -G_j {\tilde{\eta}}_j    
	\\
	\dot{\tilde\alpha}_j    &=\dot{\hat\alpha}_j-\dot\alpha_j = -k_\alpha \Lambda_j({\zeta}_j)(\eta_j+R_j\zeta_j)                 \nonumber\\
	&=-k_\alpha \Lambda_j({\zeta}_j)(\tilde\eta_j+R_j\tilde\zeta_j)
	\end{align}
\end{subequations}
Define $\rho_{qj}:=\tilde A_j(\alpha_j)\varphi_j - \tilde A_j(\hat\alpha_j){\zeta}_j$. Then the closed-loop system is converted into
\begin{subequations}
	\label{closed-loop2}
	\begin{align}
	\dot{\tilde \theta}_{ij} & =  \omega_i- \omega_j
	\label{closed2_1}	\\
	\begin{split}
	\dot \omega_j   &=  \frac{1}{M_j}\bigg(\overline P_j^{in}+ \rho_q + \sum\nolimits_{i: i\rightarrow j} \! B_{ij} {\tilde \theta}_{ij}
	-  \sum\nolimits_{k: j\rightarrow k} \! B_{jk}{\tilde \theta}_{jk}  \\
	& \ \ \ \ \ - \mu_j/\beta_j - 	d_j\omega_j  -D_j \omega_j\big) 
	\end{split}
	\\
	\label{closed2_2}
	\dot\mu_j      &= -\sum\nolimits_{k\in N_{cj}}(\mu_j-\mu_k) + \omega_j/\beta_j \\
	\label{closed2_3}
	\begin{split}
	\dot{\tilde{\eta}}_j   & =   -{\tilde{\eta}}_j  +\overline P_j^{in}+\rho_q + \sum\nolimits_{i: i\rightarrow j} \! B_{ij} {\tilde \theta}_{ij} -  \sum\nolimits_{k: j\rightarrow k} \! B_{jk}{\tilde \theta}_{jk} 
	\\
	 & \ \ \ \ \  -D_j \omega_j  - \mu_j/\beta_j - d_j\omega_j 
	 \end{split}
	 \\
	\label{closed2_4} 
	\dot{\tilde{\zeta}}_j & =   (A_j(\hat\alpha_j)-A_j(\alpha_j)){\zeta}_j + (A_j(\alpha_j)-G_jR_j){\tilde{\zeta}}_j 
	-G_j {\tilde{\eta}}_j   
	\\
	\label{closed2_5}
	\dot{\tilde\alpha}_j  & =   -k_\alpha \Lambda_j({\zeta}_j)(\tilde\eta_j+R_j\tilde\zeta_j)
	\end{align}
\end{subequations}

{\color{black}The new closed-loop system \eqref{closed-loop2} is equivalent to the original one , i.e. \eqref{system model} and \eqref{controller_varying}. We can hence analyze the equilibrium point and stability of the equivalent system \eqref{closed-loop2}.}

\section{Equilibrium  Point and Stability}
In this section, we analyze the equilibrium and stability of the closed-loop system \eqref{closed-loop2} when the noise $w_j(t)$ is \emph{not} considered.
%It is proved that the controller \eqref{controller_varying} can solve \textbf{P1}, \textbf{P2} simultaneously. We first introduce the closed-loop system composed of feedback controller and physical dynamics. 
%Its equilibrium point and stability are analyzed. 

\subsection{Equilibrium Point}
\label{subsec:optimality.1}
First we  define the  equilibrium point of the closed-loop system \eqref{closed-loop2}.
\begin{definition}
	\label{def:ep.1}
	A point $( \tilde\theta^*, \omega^*, \mu^*, \tilde\eta^*,\tilde\zeta^*, \tilde \alpha^*)$ \footnote{Given a collection of $y_i$ for $i$ in a certain set $Y$, $y$ denotes the column vector
		$y := (y_i, i\in Y)$ of a proper dimension with $y_i$ as its components.}
	is an \emph{equilibrium point} or an \textit{equilibrium} of the closed-loop system \eqref{closed-loop2} if the right-hand side of \eqref{closed-loop2} vanishes at $( \tilde\theta^*, \omega^*, \mu^*, \tilde\eta^*,\tilde\zeta^*, \tilde \alpha^*)$.
\end{definition}

%\begin{definition}
%	A point $(P^{g*}_j, \mu_j^*)$ is \emph{primal-dual optimal} if $P^{g*}_j$ is optimal 
%	for \eqref{eq:opt.1} and $\mu_j^*$ is optimal for its dual problem.
%\end{definition}
The next theorem shows that two problems \textbf{P1} and \textbf{P2} are solved simultaneously at the equilibrium.
\begin{theorem}
	\label{thm:optimality}
	At the equilibrium of \eqref{closed-loop2}, the following assertions are true.
	\begin{enumerate}
		\item $\tilde\eta_j^*=\tilde\zeta_j^*=\tilde\alpha_j^*=0$, which implies that $q_j(t)$ is accurately estimated. 
		\item System frequency is restored to its nominal value, i.e. $\omega_j^*=0$ for all $j \in N $.
		\item The marginal controllable load costs satisfy $\mu^*_j =\mu^*_k$ for all $j, k \in N$.
%		\item The equilibrium point $(\tilde\theta^*, \omega^*, \mu^*, \tilde\eta^*,\tilde\zeta^*, \tilde \alpha^*)$ is unique, with $\tilde\theta^*$ being unique up to (equilibrium) reference angles $\tilde\theta_0$.
	\end{enumerate}
\end{theorem}

\begin{proof}[Proof of Theorem \ref{thm:optimality}]
	In an equilibrium, we have
	\begin{subequations}
		\label{equilibirum1}
		\begin{align}
			0 & =  \omega_i^*- \omega_j^*
			\label{equilibirum1_1}	\\
			\label{equilibirum1_2.0}
			0 &= \overline P_j^{in}+ \rho_q^* + \sum_{i: i\rightarrow j}B_{ij} {\tilde \theta}^*_{ij}
			-  \sum_{k: j\rightarrow k}B_{jk}{\tilde \theta}^*_{jk} - \frac{\mu_j^*}{\beta_j} - 	(d_j+D_j)\omega_j^*  \\
			\label{equilibirum1_2.1}
			0 &=-\sum\nolimits_{k\in N_{cj}}(\mu_j^*-\mu_k^*) + \omega_j^*/\beta_j\\		
			\label{equilibirum1_2}
			\begin{split}
			0 &=-{\tilde{\eta}}_j^* +\overline P_j^{in}+ \rho_q^* + \sum\nolimits_{i: i\rightarrow j}B_{ij} {\tilde \theta}^*_{ij}
			-  \sum\nolimits_{k: j\rightarrow k}B_{jk}{\tilde \theta}^*_{jk} \\
			&\ \ \ \ - {\mu_j^*}/{\beta_j} - 	(d_j+D_j)\omega_j^* 
			\end{split}\\
			\label{equilibirum1_4} 
			0 & =   (A_j(\hat\alpha_j^*)-A_j(\alpha_j)){\zeta}_j^* + (A_j(\alpha_j)-G_jR_j){\tilde{\zeta}}_j^* 
			-G_j {\tilde{\eta}}_j^*     \\
			\label{equilibirum1_5}
			0 & =   -k_\alpha \Lambda_j({\zeta}_j^*)(\tilde\eta_j^*+R_j\tilde\zeta_j^*)
		\end{align}
	\end{subequations}
	We have ${\tilde{\eta}}_j^*=0$ due to \eqref{equilibirum1_2} and \eqref{equilibirum1_2.0}. Then \eqref{matrix:A} yields 
	\bq
	\label{matrix:A2}
	A_j(\hat\alpha_j^*)-A_j(\alpha_j) &=&\left [ \begin{array}{*{20}{c}}
		\textbf{0}_{2s_j\times 1} & \textbf{0}_{2s_j}\\
		0  & \tilde\alpha_{j1}, 0, \cdots, \tilde\alpha_{js_j}, 0
	\end{array}  \right ]
	%
	%		\left [ \begin{array}{*{20}{c}}
	%		\textbf{0}_{(2s_j-1)\times 1} & \textbf{0}_{2s_j-1}\\
	%		\tilde\alpha_{j1}  & 0, \tilde\alpha_{j2}, 0, \cdots, \tilde\alpha_{js_j}, 0
	%		\end{array}  \right ]
	\eq
	and
	\begin{align}
		&A_j(\alpha_j)-G_jR_j= \left [ \begin{array}{*{20}{c}}
		\textbf{0}_{2s_j\times 1} & I_{2s_j}\\
		0  & \alpha_{j1}, 0, \cdots, \alpha_{js_j}, 0
	\end{array}  \right ] \nonumber\\
	&\qquad\qquad \qquad- [\textbf{0}^T, 1, \gamma]^T\cdot[r_{j1}, \cdots, r_{j,\bar s_j-1}, 1]\nonumber\\
	&=\left [ \begin{array}{*{20}{c}}
		\textbf{0}_{\bar s_j-2, 1} & I_{\bar s_j-2}& \textbf{0}_{\bar s_j-2, 1} \\
		-r_{j1}& -r_{j2}, \cdots, -r_{j,\bar s_j-1}& 0 \\
		-\gamma r_{i1}  & \alpha_{j1}-\gamma r_{j2}, -\gamma r_{j3}, \cdots, -\gamma r_{j,\bar s_j-1} & -\gamma
	\end{array}  \right ] 
	\end{align}
	Then the first $\bar s_j-1$ dimension of \eqref{equilibirum1_4} is rewritten as 
	\bq \label{eq:e18}
	&\underbrace{\left [ \begin{array}{*{20}{c}}
			\textbf{0}_{\bar s_j-2, 1} & I_{\bar s_j-2} \\
			-r_{j1}& -r_{j2}, \cdots, -r_{j,\bar s_j-1} 
		\end{array}  \right ]}_\Psi \left[\begin{array}{*{20}{c}} {\tilde{\zeta}}_{j1}^*\\ \vdots \\ \tilde{\zeta}_{j,\bar s_j-1}^* \end{array} \right]=\textbf{0} 
	\eq
	The first matrix in \eqref{eq:e18}, denoted by $\Psi$, is nonsingular. Hence we have $\left[{\tilde{\zeta}}_{j1}^*, \cdots, \tilde{\zeta}_{j,\bar s_j-1}^*\right]^T=\textbf{0}$. Denote $\tilde{\alpha}_j:=[\tilde{\alpha}_{j1}, \cdots, \tilde{\alpha}_{j,s_j}]^T$. Then the $\bar s_j$-th dimension of \eqref{equilibirum1_4} together with \eqref{equilibirum1_5} yield 
\begin{align}
	\Lambda^T_j({\zeta}_j)\tilde{\alpha}_j^*-\gamma\tilde{\zeta}_{j,\bar s_j}^*&\equiv 0 \nonumber\\
	\Lambda_j({\zeta}_j)\tilde{\zeta}_{j,\bar s_j}^*&\equiv\textbf{0}\nonumber
\end{align}
	This implies $\tilde{\zeta}_{j,\bar s_j}^*=0$ and $\tilde{\alpha}_j^*=\textbf{0}$. The first assertion is proved. 
	
	From the first assertion, we have 
\begin{align}
	\label{equilibrium:A}
	\rho_{qj}^*=-\tilde A_j(\hat\alpha_j^*){\zeta}_j^* + \tilde A_j(\alpha_j)\varphi_j=0
\end{align}
	From \eqref{equilibirum1_1}, we have $\omega_i^*=\omega_j^*=\omega_0$, with a constant $\omega_0$. Considering the compact form of \eqref{equilibirum1_2.1}, we have 
	\begin{align}
		\label{eq:mu}
		-L\mu^*+\omega_0\cdot \beta^{-1}=0
	\end{align}
	where $\beta^{-1}:=[\beta_1^{-1}, \cdots, \beta_n^{-1}]^T$. Multiply $\textbf{1}^T$ on both sides of \eqref{eq:mu}, and we have
%		\bq
%	L:=\left[\begin{array}{*{20}{c}} L_1\\ \vdots \\ L_n \end{array} \right], \ 
%	0=L_1+L_2+\cdots+L_{n-1}+L_n\nonumber
%	\eq
%	and $L_i$ is the $i$-{th} row of $L$.
%	Sum all rows of \eqref{eq:mu}, and we have
	\begin{align}\label{omega_0}
		-\textbf{1}^T\cdot L\mu^*+\omega_0\textbf{1}^T\cdot \beta^{-1}=0=\omega_0(\beta_1^{-1}+\cdots+\beta_n^{-1})
	\end{align}
	where $\textbf{1}$ is a vector with all elements as $1$, and the second equation is due to $\textbf{1}^T\cdot L=0$. Thus we have $\omega_0=0$ due to $\beta_j>0, \forall j$, which is the second assertion. \wang{Is this more explicit?}
	
	From \eqref{eq:mu}, we have $L\mu^*=0$. Equivalently, $\mu^*=\mu_0\cdot\textbf{1}$ with a constant $\mu_0$, implying the third assertion. 
\end{proof}

{In fact, the equilibrium $(\tilde\theta^*, \omega^*, \mu^*, \tilde\eta^*,\tilde\zeta^*, \tilde \alpha^*)$ is unique, with $\tilde\theta^*$ being unique up to reference angles $\theta_0$. As the optimization problem \eqref{eq:opt.1} is with a strongly convex objective function and linear constraints, its solution $P_j^l$ is unique. Then, $\mu_j^*$ is unique by \eqref{constant_2}. In Theorem \ref{thm:optimality}, we prove that $\omega_j^*=\tilde\eta_j^*=\tilde\zeta_j^*=\tilde\alpha_j^*=0$, which are also unique. If the angle of the reference node is set as a constant $\theta_0$, $\tilde\theta^*$ is also unique (see \cite[Theorem 2]{Distributed_II:Wang}). Thus, the equilibrium point of \eqref{closed-loop2} is unique.}

From the first assertion and invoking \eqref{transformation1}, we have ${\zeta}^*_j=\varphi_j, \hat\alpha^*_j=\alpha_j$, implying the variation $q_j(t)$ is accurately eliminated. Then \textbf{P2} is solved. From the third assertion, \textbf{P1} is solved. Therefore, \textbf{P1} and \textbf{P2} are solved simultaneously.

\subsection{Asymptotic stability}
In this subsection, we prove the asymptotic stability of the closed-loop system \eqref{closed-loop2} when the noise $w_j(t)$ is not considered. We start by transforming it to an equivalent form. 

Denote $\hat \eta_j:=\tilde{\eta}_j-M_j\omega_j$ and $\nu_j:=[\hat \eta_j,\  \tilde{\zeta}_j,\ \tilde{\alpha}_j]^T$. Then \eqref{closed-loop2} can be rewritten as 
\begin{subequations}
	\label{closed-loop3}
	\begin{align}
	\dot{\tilde \theta}_{ij}\ = &\  \omega_i- \omega_j
	\label{closed3_1}	\\
	\begin{split}
	\dot \omega_j       \   =&\ \frac{1}{M_j}\bigg( \overline P_j^{in}+ \rho_{qj} + \sum\nolimits_{i: i\rightarrow j} \! B_{ij} {\tilde \theta}_{ij}
	-  \sum\nolimits_{k: j\rightarrow k} \! B_{jk}{\tilde \theta}_{jk}  \\
	 &\ \  - \mu_j/\beta_j - 	d_j\omega_j  -D_j \omega_j \big) 
	 \end{split}
	 \\
	\label{closed3_2}
	\dot{\mu}_j   \   =&\ -\sum\nolimits_{k\in N_{cj}}(\mu_j-\mu_k) + \omega_j/\beta_j
	\\
	\label{closed3_3}
	\dot{\nu} _j      \     =&\  \phi_j (\nu_j, \omega_j) 
	\end{align}
\end{subequations}
where 
\begin{align}
	\phi_j (\nu_j, \omega_j)=\left [ 
	\begin{array}{c}
		- \hat{\eta}_j-M_j\omega_j\\
		\left(\begin{array}{l}
			(A_j(\hat\alpha_j)-A_j(\alpha_j)){\zeta}_j  -G_j ({\hat{\eta}}_j+M_j\omega_j) \\
			+ (A_j(\alpha_j)-G_jR_j){\tilde{\zeta}}_j\end{array}\right)\\ 
		-k_\alpha \Lambda_j({\zeta}_j)({\hat{\eta}}_j+M_j\omega_j+R_j\tilde\zeta_j)
	\end{array}  \right ]  \nonumber
\end{align}
It is obvious that if \eqref{closed-loop2} is stable, \eqref{closed-loop3} is also stable. Thus, we turn to prove the stability of \eqref{closed-loop3}.

Consider the subsystem $\nu_j$ , we have the following Lemma. 
\begin{lemma}
	\label{lemma:subsystem_mu}
	Consider the subsystem \eqref{closed3_3} and let $\omega_j\equiv 0$. Then for each $j\in N$, there exists a $C^1$ function $U_j(t,\nu_j)$ such that 
	\begin{align}
	\begin{split}
		&\underline U_j(\nu_j)\le U_j(t,\nu_j) \le  \overline U_j(\nu_j)  \\
		&\frac{\partial U_j(t,\nu_j)}{\partial t} + \frac{\partial U_j(t,\nu_j)}{\partial \nu_j} \phi_j(\nu_j, 0) \le -\left\|\nu_j\right\|^2\\
		&\left\|\frac{\partial U_j(t,\nu_j)}{\partial \nu_j}\right\|\le b_{j0} (\left\|\nu_j\right\|+\left\|\nu_j\right\|^3)
	\end{split}
	\end{align}
	for some constant $b_{j0}>0$ and positive definite and radially unbounded functions $\underline U_j(\nu_j), \overline U_j(\nu_j)$.
\end{lemma}
The proof of Lemma \ref{lemma:subsystem_mu} is similar to \cite[Lemma 3]{Xinghu:Distributed}, which is omitted here. 

Before giving the stability result, we first study the Euclidean norm of $\left\|\rho_{qj} \right\|$ and $\left\|\frac{\partial U_j(t,\nu_j)}{\partial \nu_j}\left(\phi_j (\nu_j, \omega_j)-\phi_j (\nu_j, 0)\right)\right\|$. 
For $\rho_{qj}$, 
\begin{align}
	&\left\|\rho_{qj} \right\|=\left\|\tilde A_j(\alpha_j){\varphi_j}-\tilde A_j(\hat\alpha_j)({\tilde\zeta_j}+\varphi_j)\right\| \nonumber\\
	& \le\left\| R_j\hat A_j(\tilde\alpha_j){\varphi_j}\right\| + \left\|\tilde A_j(\tilde\alpha_j+\alpha_j){\tilde\zeta_j}\right\|  \nonumber\\
	& \le \left\| R_j\right\| \left\| \hat A_j(\tilde\alpha_j)\right\| \left\| {\varphi_j}\right\| + \left\|R_j{\tilde\zeta_j} + R_jA_j(\tilde\alpha_j+\alpha_j){\tilde\zeta_j}\right\| \nonumber\\
	& \le \left\| R_j\right\| \left\| \hat A_j(\tilde\alpha_j)\right\| \left\| {\varphi_j}\right\|  + \left\|R_j{\tilde\zeta_j}\right\| + \left\|  R_jA_j(\alpha_j){\tilde\zeta_j}\right\| +\left\|  R_j\bar A_j(\tilde\alpha_j){\tilde\zeta_j} \right\| \nonumber\\
	& \le c_2(\left\|\nu_j\right\|+\left\|\nu_j\right\|^2)
\end{align}
where 
\begin{align}
	\hat A_j(\tilde\alpha_j)&= \left[ {\begin{array}{*{20}{c}}
			\tilde\alpha_{j1}&   0& \tilde\alpha_{j2}& 0& \cdots& \tilde\alpha_{js_j}& 0
	\end{array}} \right]\nonumber\\
	\bar A_j(\tilde\alpha_j)&= \left [ \begin{array}{*{20}{c}}
		\textbf{0}_{(\bar s_j-1)\times 1} & \textbf{0}_{\bar s_j-1,\bar s_j-1}\\
		\tilde\alpha_{j1}  & 0, \tilde\alpha_{j2}, 0, \cdots, \tilde\alpha_{js_j}, 0
	\end{array}  \right ] \nonumber\\
	c_2&\ge  \left\| R_j\right\| \left\| \hat A_j(\tilde\alpha_j)\right\|+\left\|R_j\right\|+\left\|  R_jA_j(\alpha_j)\right\|, \forall j\in N \nonumber
\end{align}
The last ``$\le$''  is due to the boundedness of  $\phi_j$. Define a set $\Omega_{\nu}:=\left\{\nu |\ \sum\nolimits_{j\in N}U_j(t,\nu_j) \le \tilde c \right\}$.  Since $U_j(t,\nu_j)$ is radially unbounded, there exists a constant $\overline c$ such that $\left\|v_j(t) \right\|\le \overline c$ for any $\nu\in\Omega_{\nu}$.  In $\Omega_{\nu}$, we have 
\begin{align}
	\label{norm_rho}
	\left\|\rho_{qj} \right\|=\left\|\tilde A_j(\alpha_j){\phi_j}-\tilde A_j(\hat\alpha_j){\zeta_j}\right\|\le c_3\left\|\nu_j\right\| 
\end{align}
for a suitable $c_3>0$ (defined in \eqref{c3_first}). 

Similarly,
\begin{align}
	\label{phi1}
	&\left\|\phi_j (\nu_j, \omega_j)-\phi_j (\nu_j, 0)\right\| =\left\| \left [ 
	\begin{array}{c}
		-M_j\omega_j\\
		-G_j M_j\omega_j\\ 
		-k_\alpha \Lambda_j({\zeta}_j)M_j\omega_j
	\end{array}  \right ] \right\|  \nonumber\\
	&\le \left(\left\| M_j\right\|+\left\| G_j\right\|\left\| M_j\right\| +k_{\alpha} \left\| M_j\right\|\left\| \nu_j\right\|\right) \left\| \omega_j\right\|  \le c_3\left\|\omega_j\right\|
\end{align}
From Lemma \ref{lemma:subsystem_mu}, we have
\begin{align}
	\label{norm_U}
	\left\|\frac{\partial U_j(t,\nu_j)}{\partial \nu_j}\right\|\le c_3\left\|\nu_j\right\|
\end{align}
Combining \eqref{phi1} and \eqref{norm_U}, we have
\begin{align}
	\label{Phi}
	&\left\|\frac{\partial U_j(t,\nu_j)}{\partial \nu_j}\left(\phi_j (\nu_j, \omega_j)-\phi_j (\nu_j, 0)\right)\right\|\le \frac{1}{2}\left\|\nu_j\right\|^2+\frac{1}{2}c_3^4\left\|\omega_j\right\|^2
\end{align}   
where
\begin{align}
	\label{c3_first}
	\begin{split}
	c_3&\ge \max \big\{1,\ c_2(1+\overline c),\ b_{j0}(1+\overline c^2), \\
	&\qquad\qquad\left\| M_j\right\|+\left\| G_j\right\|\left\| M_j\right\| +k_{\alpha} \left\| M_j\right\|\overline c \big\},\ \forall j\in N
	\end{split}
\end{align}
We make an assumption.
\begin{itemize}
	\item[\textbf{A1}:] The control parameter $d_j$ satisfies 
\end{itemize}
\begin{align}
\label{eq:dj}
d_j>\max \left\{\frac{1+2c_3^6}{2}-D_j,\ \frac{2c^2_3+1}{4c^2_3}+\frac{2c_3^6-c_3^4}{2c_3^2-2}+2c_3^2-D_j\right\}
\end{align}

A1 is easy to satisfy by letting $d_j$  large enough.
Denote the state variables of \eqref{closed-loop3} as $x=\left[\tilde{\theta}^T, \omega^T, {\mu}^T, \nu^T \right ]^T$ and $x_1=\big[\tilde{\theta}^T, \omega^T,$ ${\mu}^T \big ]^T$. Similar to Definition \ref{def:ep.1}, we have 
\begin{definition}
	\label{def:ep.2}
	A point $x^*$ is an \emph{equilibrium point} of the closed-loop system \eqref{closed-loop3} if the right-hand side of \eqref{closed-loop3} vanishes at $x^*$.
\end{definition}

Define a Lyapunov candidate function as
\begin{align}
	\label{Lyapunov}
	V(t,x_1,\nu)=\frac{1}{2c_3^2}V_1+V_2
\end{align}
where 
\begin{align}
	\label{Lyapunov2}
	V_1=\frac{1}{2}(x_1-x_1^*)^T\Gamma(x_1-x_1^*)
\end{align}
with $\Gamma:=\text{diag}\left(B, M, I_n \right)$,
\begin{align}
	\label{Lyapunov3}
	V_2=\sum\nolimits_{j\in N}U_j(t,\nu_j)
\end{align}

From Lemma \ref{lemma:subsystem_mu} and \eqref{Lyapunov2}, there are positive definite and radially unbounded functions $\underline V(x_1,\nu), \overline V(x_1,\nu)$ such that $\underline V(x_1,\nu)\le V(t,x_1,\nu) \le \overline V(x_1,\nu)$. Define a set $\Omega_{\overline V}=\left\{(x_1, \nu) |\ \overline  V(x_1,\nu) \le \tilde c \right\}$. We have $\forall (x_1,\nu)\in \Omega_{\overline V}, \text{then}\ \nu\in \Omega_{\nu}$ and $\left\|v_j(t) \right\|\le \overline c$.

Finally, the stability result is given.
\begin{theorem}
	\label{Theorem:stability}
	Assume A1 holds. Then every trajectory of \eqref{closed-loop3} $x(t)$ starting from $\Omega_{\overline V}$ converges to $x^*$ asymptotically.
\end{theorem}
\wang{As the equilibrium point of \eqref{closed-loop2} is unique, the equilibrium point of \eqref{closed-loop3} is also unique. \eqref{closed-loop2} and \eqref{closed-loop3} are equivalent. }

\begin{proof}[Proof of Theorem \ref{Theorem:stability}]

	Define the following function
	\begin{equation}
	h( x_1) = \left[ {\begin{array}{*{20}{c}}
		BC^T \omega \\
		\overline P^{in} - \beta^{-1}\mu  -(D+d) \omega - CB\tilde\theta \\
		-L\mu+\beta^{-1}\omega   
		\end{array}} \right]  
	\end{equation}

	The derivative of $V_1$ is 
	\begin{align}
		\label{Lyapunov2_derivative}
		\dot V_1=(x_1-x_1^*)^Th( x_1)+ \sum\limits_{j\in N}{\omega_j \left(\tilde A_j(\alpha_j){\varphi_j}-\tilde A_j(\hat\alpha_j){\zeta_j}\right)}
	\end{align}
	The first part of $\dot V_1$ is 
	\begin{align}
		& (x_1-x_1^*)^Th(x_1) \nonumber\\
		& =  \int_0^1 (x_1-x_1^*)^T \frac{\partial }{\partial y}h( y(s)) (x_1-x_1^*) ds + (x_1-x_1^*)^T h(x_1^*)\nonumber\\
		& \le  \frac{1}{2}\int_0^1 (x_1-x_1^*)^T \left[\frac{\partial^T }{\partial y} h( y(s)) + \frac{\partial }{\partial y} h( y(s)) \right] (x_1-x_1^*) \text{d}s \nonumber\\
		\label{dot Wk1}
		& = \int_0^1 (x_1-x_1^*)^T \left[H(y(s)) \right] (x_1-x_1^*) \text{d}s
	\end{align}
	where $y(s)=x_1^*+s(x_1-x_1^*)$. The second equation is from the fact that $h(x_1)-h(x_1^*) = \int_0^1 \frac{\partial }{\partial y}h( y(s)) (x_1-x_1^*) \text{d}s$. The inequality is due to either $h(x_1^*) =0$ or $h(x_1^*) < 0, x_1\ge 0$, i.e. $(x_1-x_1^*)^T h(x_1^*)\le0$.
	\begin{align}
		\frac{\partial h(x_1) }{\partial x_1}  = \left[ {\begin{array}{*{20}{c}}
				0&       BC^T&                0\\
				-CB&       -(D+d)&             - \beta^{-1}\\
				0&       \beta^{-1}& -L
		\end{array}} \right]
	\end{align}
	where $D=\text{diag}(D_i)$,  $d=\text{diag}(d_i)$, $C$ is the incidence matrix of the communication graph.
	
	Finally, $H$ in \eqref{dot Wk1} is
	\begin{align}
		H&=\frac{1}{2}\left[\frac{\partial^T }{\partial x_1} h(x_1)  + \frac{\partial }{\partial x_1} h(x_1) \right]  =	\left[ {\begin{array}{*{20}{c}}
				0&       0&                0\\
				0&       -(D+d)&           0\\
				0&       0&                -L
		\end{array}} \right]\nonumber
	\end{align}
	The second part of $\dot V_1$ is
	\begin{align}
		&\sum\nolimits_{j\in N}{\omega_j \left(\tilde A_j(\alpha_j){\phi_j}-\tilde A_j(\hat\alpha_j){\zeta_j}\right)}\nonumber\\
		&\quad \le \frac{1}{2}\left\|\omega\right\|^2 + \frac{1}{2}\sum\nolimits_{j\in N}(\tilde A_j(\alpha_j){\varphi_j}-\tilde A_j(\hat\alpha_j){\zeta_j})^2\nonumber\\
		&\quad \le \frac{1}{2}\left\|\omega\right\|^2 + \frac{1}{2}c^2_3\left\|\nu\right\| 
	\end{align}
	where the last inequality is due to \eqref{norm_rho}.
	
	Thus, 
	\begin{align}
		\label{V1_dot}
		\dot V_1\le & \int_0^1 (x_1-x_1^*)^T \left[H(y(s)) \right] (x_1-x_1^*) \text{d}s + \frac{1}{2}\left\|\omega\right\|^2 + \frac{1}{2}c^2_3\left\|\nu\right\| 
	\end{align}
	
	The derivative of $V_2$ is
	\begin{align}
		\label{Lyapunov3:derivative}
		\dot V_2&=\sum\nolimits_{j\in N}\left(\frac{\partial U_j(t,\nu_j)}{\partial t} + \frac{\partial  U_j(t,\nu_j)}{\partial \nu_j} \phi_j(\nu_j, \omega_j)\right) \nonumber\\
		&=\sum\nolimits_{j\in N}\left(\frac{\partial U_j(t,\nu_j)}{\partial t} + \frac{\partial  U_j(t,\nu_j)}{\partial \nu_j} \phi_j(\nu_j, 0)\right)\nonumber\\
		&\quad +\sum\nolimits_{j\in N}\left( \frac{\partial  U_j(t,\nu_j)}{\partial \nu_j} \left(\phi_j(\nu_j, \omega_j)-\phi_j(\nu_j, 0)\right)\right) \nonumber\\
		&\le -\left\|\nu \right\|^2+\frac{1}{2}\left\|\nu\right\|^2+\frac{1}{2}c_3^4\left\|\omega\right\|^2\nonumber\\
		&=-\frac{1}{2}\left\|\nu\right\|^2+\frac{1}{2}c_3^4\left\|\omega\right\|^2
	\end{align}
	where the inequality is due to Lemma \ref{lemma:subsystem_mu} and \eqref{Phi}.

	In $\Omega_{\overline V}$, the derivative of $V$ is 
	\begin{align}
	\label{Lyapunov_derivative}
	\dot V
	&\le \frac{1}{2c_3^2}\int_0^1 (x_1-x_1^*)^T \left[H(y(s)) \right] (x_1-x_1^*) \text{d}s + \frac{1}{4c_3^2}\left\|\omega\right\|^2 \nonumber\\
	& + \frac{1}{4c_3^2}\sum\nolimits_{j\in N}(\tilde A_j(\alpha_j){\phi_j}-\tilde A_j(\hat\alpha_j){\zeta_j})^2 -\frac{1}{2}\left\|\nu\right\|^2+\frac{1}{2}c_3^4\left\|\omega\right\|^2 \nonumber\\
	&\le -\frac{1}{4}\left\|\nu\right\|^2 +\frac{1}{2c_3^2}\int_0^1 (x_1-x_1^*)^T \left[ H(y(s)) \right] (x_1-x_1^*) \text{d}s \nonumber\\
	&+ \frac{1+2c_3^6}{4c_3^2}\left\|\omega\right\|^2   
	\end{align}
	Define $\tilde H$ as 
	$\tilde H:=\left[ {\begin{array}{*{20}{c}}
			0&       0&                0\\
			0&       -(D+d)+\frac{1+2c_3^6}{2}I_n&           0\\
			0&       0&                -L
	\end{array}} \right]. $
	
	Then we have
	\bq
	\dot{V}& \le& -\frac{1}{4}\left\|\nu\right\|^2 +\frac{1}{2c_3^2}\int_0^1 (x_1-x_1^*)^T \; \tilde H \;  (x_1-x_1^*) \text{d}s
	\eq
	
	It is obvious that $\tilde{H}\le0$ holds if 
\begin{align}
	\label{d_j}
	-(D+d)+\frac{1+2c_3^6}{2}I_n&<0
\end{align}
	where $I_n$ is an $n$-dimensional identity matrix. Indeed, Assumption A1 guarantees that  \eqref{d_j} holds.
	
	By LaSalle's invariance principle, we can prove that the trajectory $x(t)$ converges to the largest invariant subset of 	
	\begin{align}
		W_1=\left\{x| \nu^*=0, \omega=\omega^*=0, \mu=\mu^*  \right\}.\nonumber
	\end{align}
	
	Next we will prove that  the convergence is to an equilibrium point. Since $\omega=\omega^*$ are constants, $\tilde{\theta}=C^T\omega^*$ are also constants. Then by \cite[Corollary 4.1]{Khalil:Nonlinear}, $x(t)$ will converge to its equilibrium point $x^*$ asymptotically.	
\end{proof}

%The detailed proof of Theorem \ref{Theorem:stability} is given in the appendix due to the space limit. Here a proof outline is given. First we give two parts of Lyapunov function
%\begin{align}
%	\label{Lyapunov2}
%	V_1=\frac{1}{2}(x_1-x_1^*)^T\Gamma(x_1-x_1^*)
%\end{align}
%where $\Gamma:=\text{diag}\left(B, M^{-1}, I_n \right)$.
%\begin{align}
%	\label{Lyapunov3}
%	V_2=\sum\nolimits_{j\in N}U_j(t,\nu_j)
%\end{align}
%Then the Lyapunov function is 
%\begin{align}
%	\label{Lyapunov}
%	V(x_1,\nu)=\frac{1}{2c_3^2}V_1+V_2
%\end{align}
%where $c_3$ is a constant, whose value is given in the Appendix.
%
%It is proved that $\dot V \le 0$ and $(x_1,\nu)$ is bounded. Then by LaSalle's invariance principle, we can prove $x(t)$ will converge to its equilibrium point $x^*$ asymptotically.

\section{Robustness Analysis}
\label{Section Robust}
\subsection{Robustness Against Uncertain Parameter $D_j$}
In the controller \eqref{controller_varying}, the exact value of $D_j$ is difficult to know, and may even change. However, we claim that the inaccuracy of $D_j$ does not influence the equilibrium point of the closed-loop system \eqref{closed-loop2} and its stability, as we explain.

We first consider the equilibrium point. Suppose the estimation  of $D_j$ is $\hat{D}_j$ and the estimation error is $\Delta D_j:=\hat{D}_j-{D}_j$. As $D_j>0$, we assume its estimation $\hat{D}_j > 0$. \wang{I think it is reasonable to assume $\hat{D}_j>0$. Although we do not know the exact value of ${D}_j$, we know it is positive. Thus, it is reasonable to assume its estimation is also positive. }
Then \eqref{closed2_3} can be rewritten as 
\begin{align}
\label{D_uncertainty}
	\begin{split}
	\dot{\tilde{\eta}}_j  & =  -{\tilde{\eta}}_j  +\overline P_j^{in}+\rho_{qj} + \sum\nolimits_{i: i\rightarrow j} \! B_{ij} {\tilde \theta}_{ij} -  \sum\nolimits_{k: j\rightarrow k} \! B_{jk}{\tilde \theta}_{jk} 
	\\
	&\qquad  -D_j \omega_j - \Delta D_j \omega_j  - \mu_j/\beta_j - d_j\omega_j
	\end{split}
\end{align}
Since $\omega_j$ vanishes at equilibrium,  $\Delta D_j$ does not influence the equilibrium point of the closed-loop system 
\eqref{closed2_1}-\eqref{closed2_2}, \eqref{D_uncertainty}, \eqref{closed2_4}-\eqref{closed2_5}. \wang{$\omega_j$ vanishes can be proved following \eqref{eq:mu} and \eqref{omega_0}. In fact, from \eqref{equilibirum1_1} and \eqref{equilibirum1_2.1}, we can prove $\omega_j^*$=0. }

Next, we discuss stability. When $\Delta {D}_j$ is considered, \eqref{closed3_3} is rewritten as 
\begin{align}
\label{new_nu}
\dot{\nu} _j  &= \left [ 
\begin{array}{c}
	- \hat{\eta}_j-(M_j+\Delta
	D_j)\omega_j\\
	\left(\begin{array}{l}
		(A_j(\hat\alpha_j)-A_j(\alpha_j)){\zeta}_j  -G_j ({\hat{\eta}}_j+M_j\omega_j) \\
		+ (A_j(\alpha_j)-G_jR_j){\tilde{\zeta}}_j\end{array}\right)\\ 
	-k_\alpha \Lambda_j({\zeta}_j)({\hat{\eta}}_j+M_j\omega_j+R_j\tilde\zeta_j)
\end{array}  \right ] 
\end{align}
Suppose $\tilde x(t)$ are state variables of $\eqref{closed3_1}-\eqref{closed3_2}, \eqref{new_nu}$, and $\tilde x^*$ is an equilibrium point of $\tilde x(t)$. 

\begin{itemize}
	\item[\textbf{A2}:] The parameter $d_j$ satisfies \eqref{eq:dj}, where $c_3$ is given by 
\end{itemize}
\bq
c_3&\ge& \max \big\{1,\ c_2(1+\overline c),\ b_{j0}(1+\overline c^2), \nonumber\\
&&\qquad\quad\left\| M_j+\Delta D_j\right\|+\left\| G_j\right\|\left\| M_j\right\| +k_{\alpha} \left\| M_j\right\|\overline c \big\}. \nonumber
\eq
We have the following result.

\begin{corollary}
	\label{lemma_uncetainty_converg}
	Assume A2 holds, every trajectory  $\tilde x(t)$  of  \eqref{closed2_1}-\eqref{closed2_2}, \eqref{D_uncertainty}, \eqref{closed2_4}-\eqref{closed2_5} starting from $\Omega_{\overline V}$ converges to the equilibrium point $\tilde x^*$ asymptotically.
\end{corollary}
Note that one can always choose a large enough $d_j$. Hence Corollary \ref{lemma_uncetainty_converg} can be easily proved following the same proof of Theorem \ref{Theorem:stability}, which is omitted here. 

In summary,  the unknown parameter $D_j$ does not influence the equilibrium point and its stability,indicating that our controller is robust against parameter uncertainty.  
\subsection{Robustness Against Unknown  Disturbances $w_j(t)$}
To attenuate the effect of $w_j(t)$, one needs to guarantee that, for a given constant $\gamma>0$, the robust performance index $\left\|\omega_j(t)\right\|^2\le\gamma\left\|w_j(t)\right\|^2$ holds. \cite[Chapter 16]{zhou1996robust}, \cite{Qin2017Input}. It means that, for a bounded external disturbance $w_j(t)$, the frequency deviation is always bounded by the given $\gamma$. A smaller $\gamma$ results in a better attenuation performance. The lower bound of $\gamma$ (if it exists) is referred to as $L_2$ gain of the system.

When considering $w_j(t)$, the closed-loop system is
\begin{subequations}
	\label{closed-loop4}
	\begin{align}
	\dot{\tilde \theta}_{ij}\  \ =&\ \  \omega_i- \omega_j
	\label{closed4_1}	\\
	\begin{split}
	\dot \omega_j       \ \  =&\ \ \frac{1}{M_j}\bigg( \overline P_j^{in} + w_j(t)+ \rho_{qj} + \sum_{i: i\rightarrow j} \! B_{ij} {\tilde \theta}_{ij}
	-  \sum_{k: j\rightarrow k} \! B_{jk}{\tilde \theta}_{jk}  \\
	 &\ \ - \mu_j/\beta_j - 	d_j\omega_j  -D_j \omega_j \big) 
	 \end{split}
	 \\
	\label{closed4_2}
	\dot{\mu}_j  \  \  =&\ \ -\sum\nolimits_{k\in N_{cj}}(\mu_j-\mu_k) + \omega_j/\beta_j
	\\
	\label{closed4_3}
	\dot{\nu} _j      \ \    =&\ \ \tilde\phi_j (\nu_j, \omega_j, w_j)
	\end{align}
\end{subequations}
where
\bq
%\label{new_nu2}
\tilde \phi_j &= \left [ 
\begin{array}{c}
	- \hat{\eta}_j-M_j\omega_j-w_j(t)\\
	\left(\begin{array}{l}
		(A_j(\hat\alpha_j)-A_j(\alpha_j)){\zeta}_j  -G_j ({\hat{\eta}}_j+M_j\omega_j) \\
		+ (A_j(\alpha_j)-G_jR_j){\tilde{\zeta}}_j\end{array}\right)\\ 
	-k_\alpha \Lambda_j({\zeta}_j)({\hat{\eta}}_j+M_j\omega_j+R_j\tilde\zeta_j)
\end{array}  \right ]. \nonumber
\eq

By an analysis similar to \eqref{phi1}, we have 
\begin{align}
\label{phi2}
\left\|\phi_j (\nu_j, \omega_j, w_j)-\phi_j (\nu_j, 0, 0)\right\| 
%	=\left\| \left [ 
%	\begin{array}{c}
%	-M_j\omega_j-w_j\\
%	-G_j M_j\omega_j\\ 
%	-k_\alpha \Lambda_j({\zeta}_j)M_j\omega_j
%	\end{array}  \right ] \right\|  \nonumber\\
%	&\le \left(\left\| M_j\right\|+\left\| G_j\right\|\left\| M_j\right\| +k_{\alpha} \left\| M_j\right\|\left\| \nu_j\right\|\right) \left\| \omega_j\right\| + \left\| w_j\right\|  \nonumber\\
\le c_3\left\|\omega_j\right\| + \left\| w_j\right\|
\end{align}
where  $c_3$ is same as that in \eqref{c3_first}. Then 
\begin{align}
	\label{Phi2}
	\begin{split}
	&\left\|\frac{\partial U_j(t,\nu_j)}{\partial \nu_j}\left(\phi_j (\nu_j, \omega_j, w_j)-\phi_j (\nu_j, 0, 0)\right)\right\| \\
	&\qquad\qquad\qquad \le \frac{1}{2}\left\|\nu_j\right\|^2  +\frac{2c_3^6-c_3^4}{2c_3^2-2}\left\|\omega_j\right\|^2 +\frac{1}{2}\left\|w_j\right\|^2
	\end{split}
\end{align}
Using  $V_1$, $V_2$  defined in \eqref{Lyapunov2} and \eqref{Lyapunov3} again,  we have
\begin{align}
%\label{V1_dot}
\begin{split}
	\dot V_1\ \le &\ \int_0^1 (x_1-x_1^*)^T \left[H(y(s)) \right] (x_1-x_1^*) \text{d}s + \frac{1}{2} c^2_3\left\|\nu\right\|^2 \\
	\ &\ + \frac{1}{2}\left\|\omega\right\|^2 + \frac{1}{4c^2_3}\left\|\omega\right\|^2 + c^2_3\left\|w\right\|^2
\end{split}
\end{align}
and
\begin{align}
%	\label{Lyapunov3:derivative}
\dot V_2&\le -\left\|\nu \right\|^2+\frac{1}{2}\left\|\nu_j\right\|^2+\frac{2c_3^6-c_3^4}{2c_3^2-2}\left\|\omega\right\|^2 +\frac{1}{2}\left\|w\right\|^2 \nonumber\\
&=-\frac{1}{2}\left\|\nu\right\|^2+\frac{2c_3^6-c_3^4}{2c_3^2-2}\left\|\omega\right\|^2 +\frac{1}{2}\left\|w\right\|^2
\end{align}
Using the same Lyapunov function as in \eqref{Lyapunov} gives 
\begin{align}
	\begin{split}
	\dot V\ \le&\ -\frac{1}{4}\left\|\nu\right\|^2 -\frac{1}{2c_3^2}(\mu-\mu^*)^TL(\mu-\mu^*)+ \left\|w\right\|^2\\
	\ &\ -  \frac{1}{2c_3^2}\omega^T \left(D+d-\frac{2c^2_3+1}{4c^2_3}I_n-\frac{2c_3^6-c_3^4}{2c_3^2-2}I_n\right)\omega 
	\end{split}
\end{align}
Thus, we have 
\begin{align}
\label{robust1}
\left\|\omega_j\right\|^2\le \gamma \left\|w_j\right\|^2
\end{align}
where
\begin{align}
	\label{robust2}
	\frac{1}{\gamma}&=\min\left\{\frac{1}{2c_3^2}\left(D_j+d_j-\frac{2c^2_3+1}{4c^2_3}-\frac{2c_3^6-c_3^4}{2c_3^2-2}\right)\right\}, \forall j\in N
\end{align}
We have  $\frac{1}{\gamma}>1$ due to Assumption A1.

Inequalities \eqref{robust1} and \eqref{robust2} indicate that the controller is robust to  $w_j(t)$ with the $L_2$-gain $\gamma<1$. In practice, the amplitudes of $w_j(t)$ are usually quite small. As a consequence, the deviation of $\omega_j$ is also small. According to \eqref{robust2},  a larger $d_j$ is helpful to  enhance the attenuation performance. 

\begin{figure}[t]
	\centering
	\includegraphics[width=0.42\textwidth]{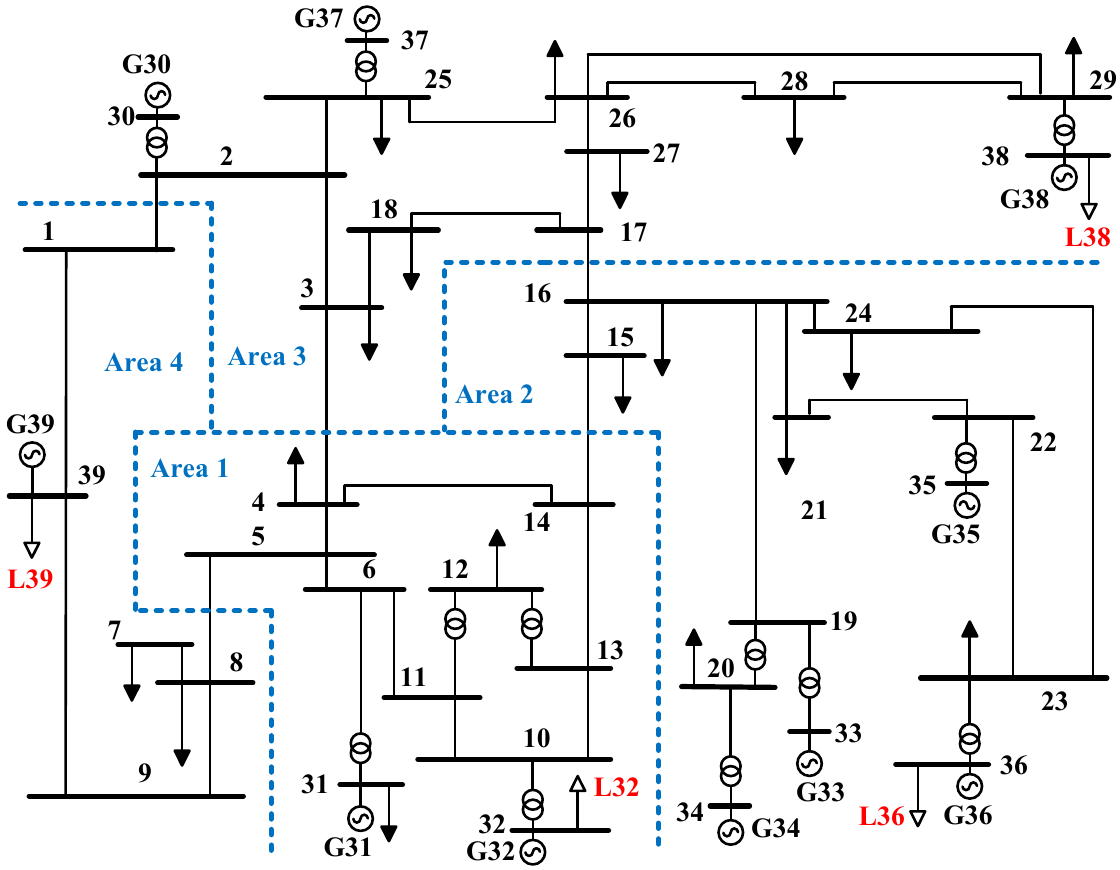}
	\caption{The New England 39-bus system}
	\label{fig:system}
\end{figure}

The analysis above shows that the controller is  robust in terms of uncertain parameter $D_j$ and unknown disturbance $w_j(t)$. Hence \textbf{P3} is resolved.

\section{Case studies}
\subsection{System Configuration}
To verify the performance of the proposed controller, the New England 39-bus system with 10 generators as shown in Figure \ref{fig:system} is used for test. All simulations are implemented in the commercial power system simulation software PSCAD. 

We add four (aggregate) controllable  loads to the system by connecting them at buses 32, 36, 38 and 39, respectively. Their initial values  are set as $(74.1,\ 52.7,\ 52.7,\ 105.4)$ MW. Then the  system is divided into four control areas, as shown in Figure \ref{fig:system}. Each area contains a controllable load. \textcolor{black}{The communication graph is undirected and set as  $L32\leftrightarrow{}{}L36\leftrightarrow{}{}L38\leftrightarrow{}{}L39\leftrightarrow{}{}L32$. For simplicity, we assume the communication is perfect with no delay or loss.}
\begin{figure}[t]
	\centering
	\includegraphics[width=0.49\textwidth]{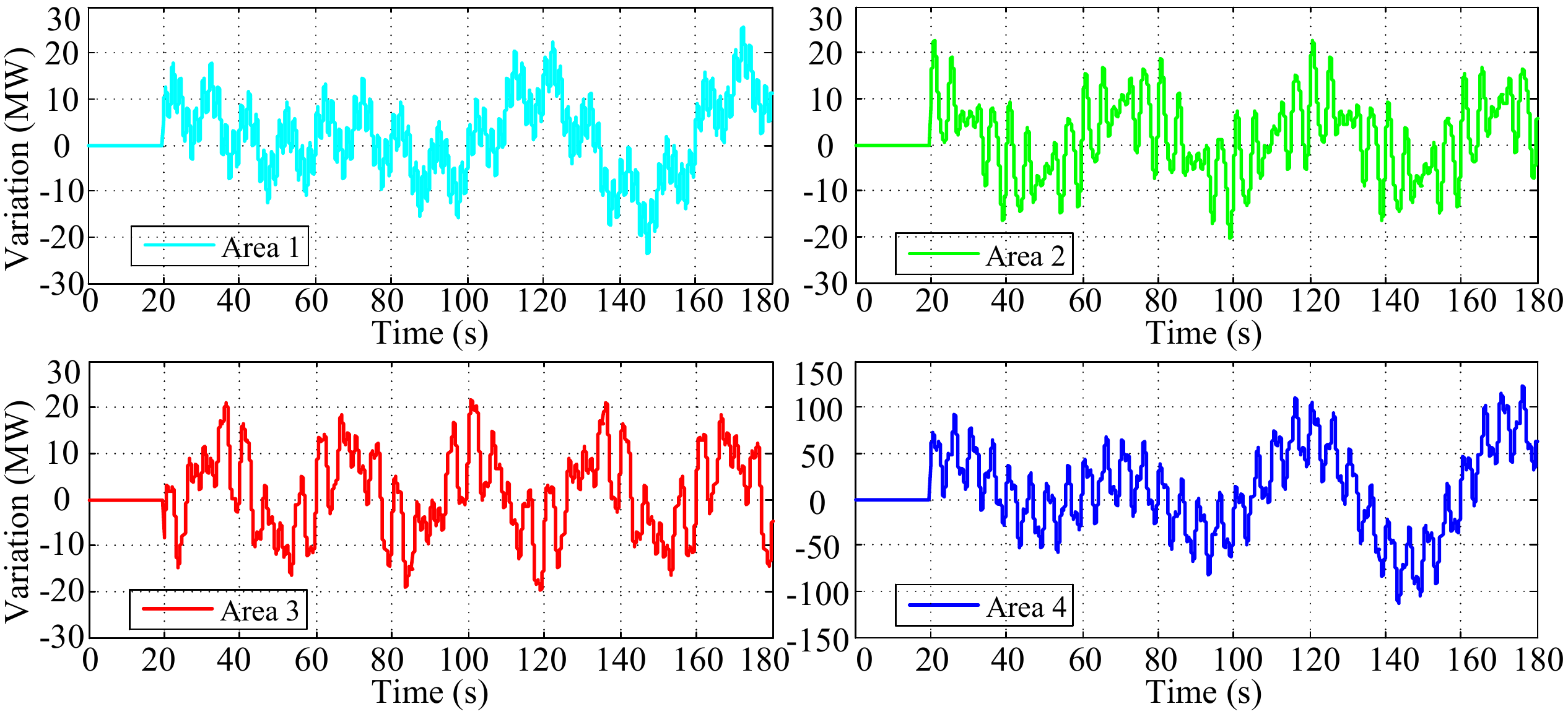}
	\caption{Power variation of renewable resources in each area}
	\label{fig:variation}
\end{figure}
In our tests, two cases are studied based on different data: 1) self-generated data in PSCAD; 2) the real data of an offshore wind farm. The variation in the first case is faster than the latter. {\color{black}Parameters used in the controller \eqref{controller_varying} are given in Table \ref{tab_DistConstraints}.
\begin{table}[h]
	\centering
	\footnotesize
	\caption{ {Parameters used in the controller \eqref{controller_varying}}}
	\label{tab_DistConstraints}
	\begin{tabular}{c c c c c}
		\hline 
		Area $i$ & 1 & 2 & 3 & 4 \\
		\hline
		$\beta$   & 1   & 0.8 & 0.8 & 0.4   \\
		$d$   & 1000   & 1000 & 1000 & 1000   \\
		$D$   & 50   & 50 & 50 & 80   \\
		$\gamma$    & 1 & 1 & 1 & 1  \\
		$k_\alpha$  & 10 & 10 & 10 & 10  \\
		\hline
	\end{tabular}
\end{table}
The value of $B_{ij}$ is given in Table \ref{tab_Bij}.
\begin{table}[h]
	\centering
	\footnotesize
	\caption{ {Parameters used in the controller \eqref{controller_varying}}}
	\label{tab_Bij}
	\begin{tabular}{c c c c c c}
		\hline 
		Line & (1, 2) & (1, 3) & (1, 4) & (2, 3) & (2, 4)\\
		\hline
		$B_{ij}$  & 46  & 47 & 89 & 112 & 24   \\
		\hline
	\end{tabular}
\end{table}
The $R_j$ used in \eqref{controller_varying} for each area is $R_j=[1, 6, 15, 20, 15, 6, 1]$. The corresponding polynomial is $(x+1)^6$, where all the roots are $-1$, satisfying the requirement. 
}

\subsection{Self-generated data}
In the first case, the varying power in each area is shown in Figure \ref{fig:variation}. Note that the functions of the four curves in Figure \ref{fig:variation} are unknown to the controllers. In the controller design, we choose $s_j=3$ in \eqref{sin_functions}. Note that this  does not mean the actual power variation (curves in  Figure \ref{fig:variation}) are superposition of only three sinusoidal functions. 

In this subsection, $\overline P_j^{in}$ in each area are $(15,\ 15,\ 15,\ 15)$ MW, which are the prediction of aggregated load. It should be pointed out that the prediction is not accurate. The offset errors are  $(1,\ 1,\ 1,\ 5)$ MW, which are relatively small but unknown. 
We compare the performances using controller \eqref{controller_constant} and \eqref{controller_varying}. Both the two controllers are applied at $t=20$s. The system frequencies  are compared in Figure \ref{fig:load_frequency}.
\begin{figure}[t]
	\centering
	\includegraphics[width=0.32\textwidth]{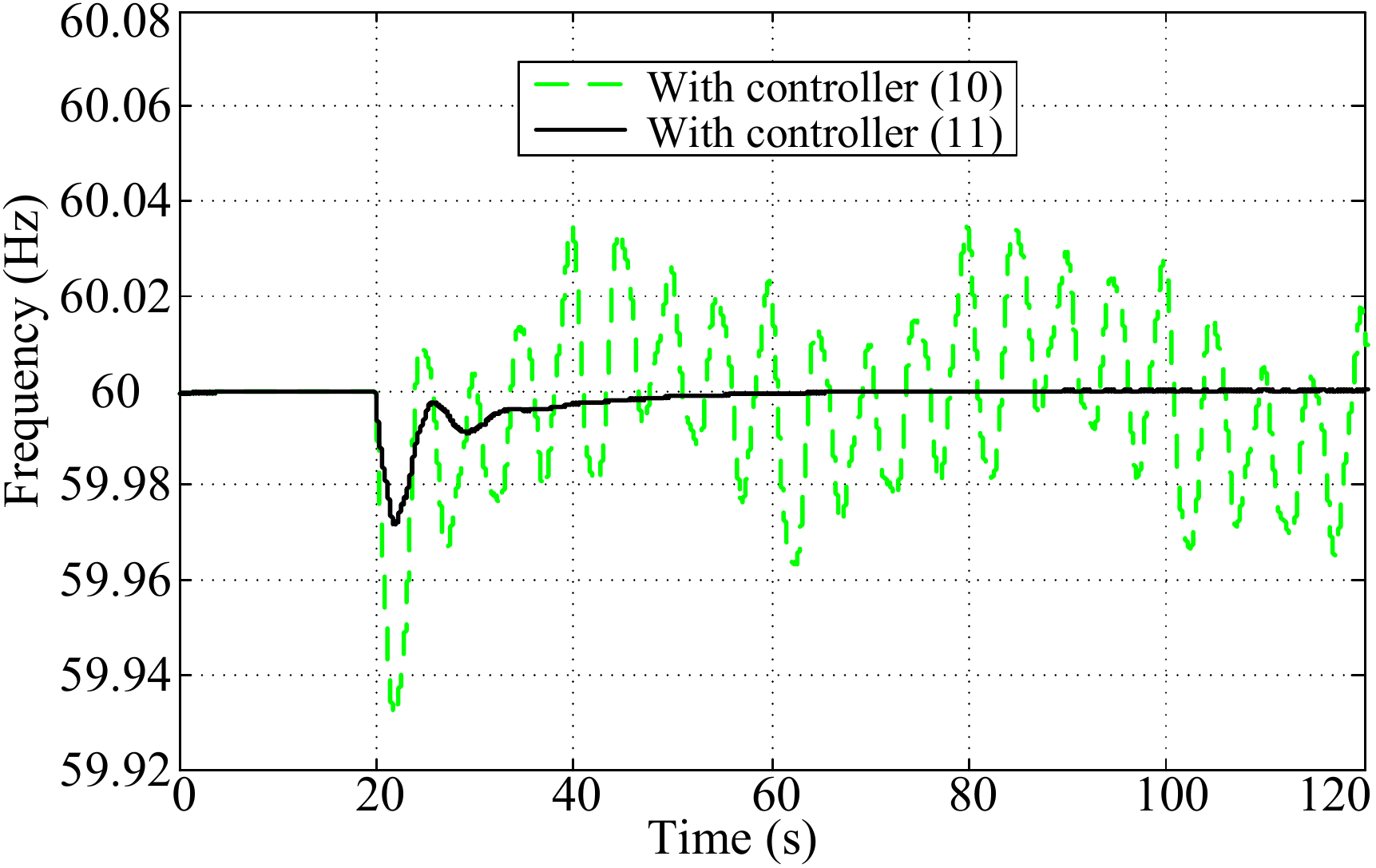}
	\caption{System frequencies under two controls}
	\label{fig:load_frequency}
\end{figure}

The green line stands for the frequency dynamics using \eqref{controller_constant}. The frequency oscillation is fierce and nadir is quite low. The black line stands for frequency dynamics using \eqref{controller_varying}. In this situation, the nominal frequency is recovered fast without oscillation. The frequency nadir is much higher than that using \eqref{controller_constant}. This result confirms  that our controller can still work well when $\overline P_j^{in}\neq0$. 

The dynamics of $\mu_j$ are given in the left part of Figure \ref{fig:mu_with_load}. The green line stands for $\mu_j$ using \eqref{controller_constant}, while the black line stands for that using \eqref{controller_varying}. $\mu_j$ of each area converges to a same value, which implies the optimality is achieved, i.e., $\overline P_j^{in}$ is balanced economically. 
%\fliu{There should be four pairs of lines, right?}
\begin{figure}[t]
	\centering
	\includegraphics[width=0.49\textwidth]{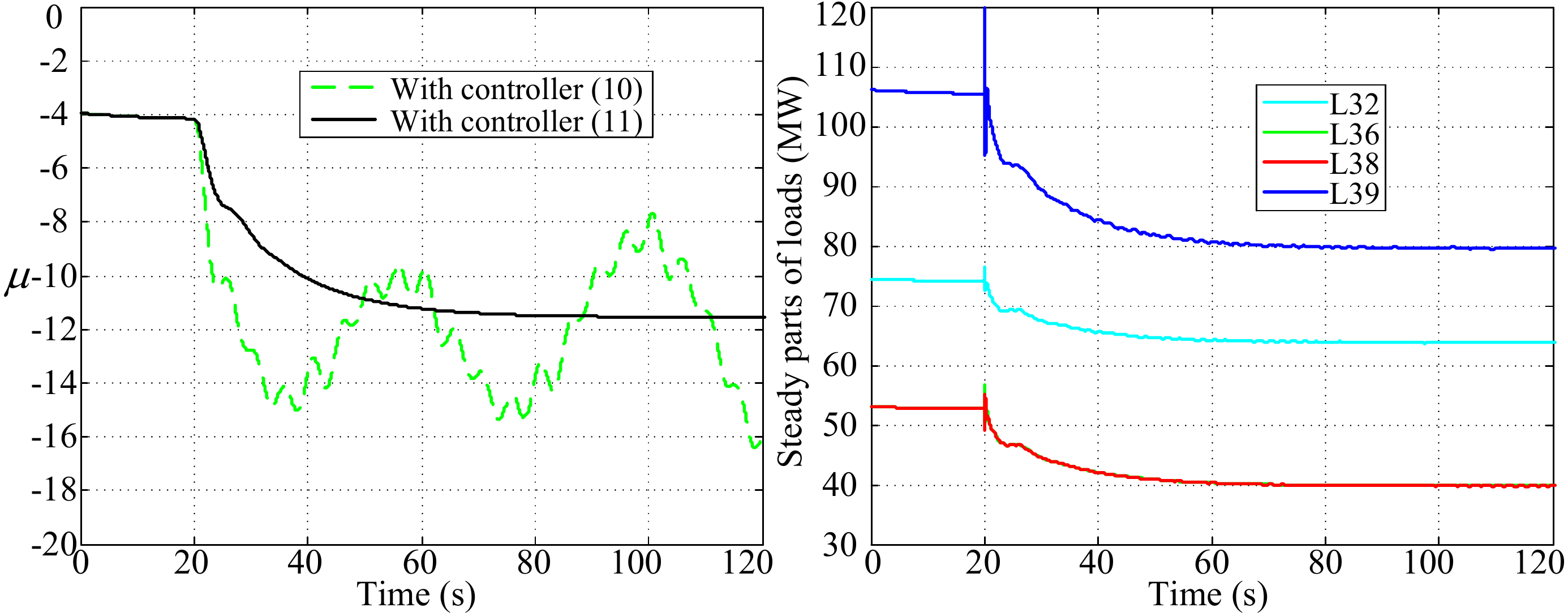}
	\caption{Dynamics of $\mu$ and steady parts of controllable loads}
	\label{fig:mu_with_load}
\end{figure}

In this scenario, the controllable load in each area is also composed of two parts: a steady part to balance $\overline P_j^{in}$ and a  variation part to mitigate the effects of $\tilde{q}_j(t)$. 
The steady part of controllable load is given in the right part of Figure \ref{fig:mu_with_load}. The controllable loads in the steady state are $(63.8,\ 39.8,\ 39.8,\ 79.6)$ MW. The result is the same as that obtained using CVX to solve the optimization counterpart (i.e., OLC problem \eqref{eq:opt.1}). 
%\begin{figure}[t]
%	\centering
%	\includegraphics[width=0.32\textwidth]{mu_l_with_load.pdf}
%	\caption{Dynamics of steady parts of controllable loads}
%	\label{fig:l_with_load}
%\end{figure}

To demonstrate  it more clearly, we define an error index $Err_j$ as below.
\begin{align}
	\label{error_index}
	Err_j:=\frac{\int_{t_0}^{t_1} \sqrt{\left(\tilde A_j(\hat\alpha_j){\zeta}_j-\tilde{q}_j(t)\right)^2} \text{d}t }{\int_{t_0}^{t_1} \sqrt{\left(\tilde{q}_j(t)\right)^2} \text{d}t }
\end{align}

The performance of  controllable load tracking  power variation in each area is given in Figure \ref{fig:load_variations_attenuation}. We can find that the controllable loads coincide to the power variations with high accuracy. 
%We also show the error between variation and internal model output, which is shown in Figure \ref{fig:error_without_load}. The error is still very small. 
Again, the error index $Err_j$ with $t_0=20$ and $t_1=120$ in this situation are $(0.0084,\ 0.0026,\ 0.0057,\ 0.0019)$, which are also very small. 
\begin{figure}[t]
	\centering
	\includegraphics[width=0.49\textwidth]{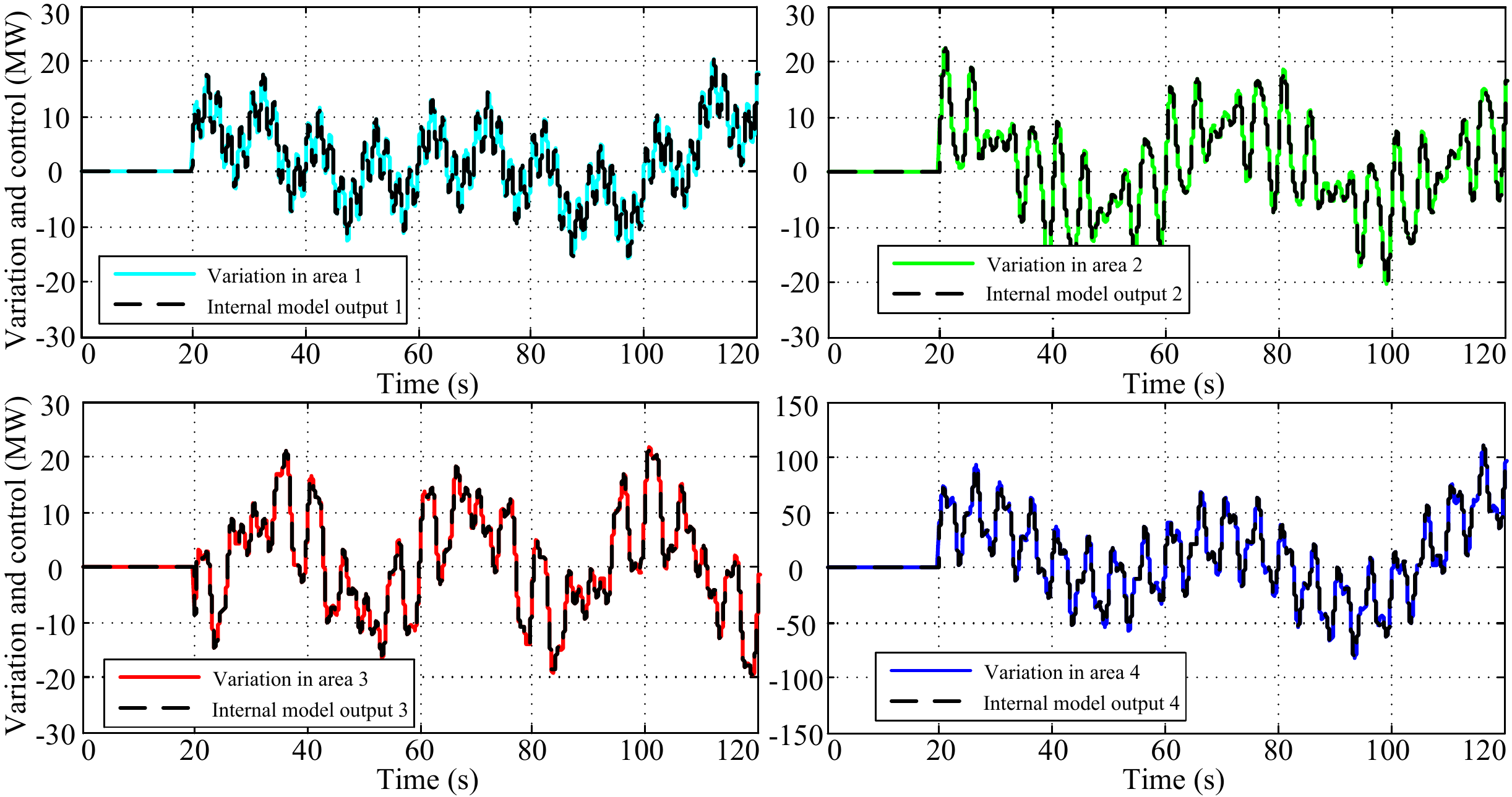}
	\caption{Variation and internal model output with load increases}
	\label{fig:load_variations_attenuation}
\end{figure}

%\begin{figure}[t]
%	\centering
%	\includegraphics[width=0.48\textwidth]{error_with_load.pdf}
%	\caption{Error between variation and internal model output}
%	\label{fig:error_with_load}
%\end{figure}

\subsection{Performance under Unknown Disturbances}
To test the performance of our controller under high-frequency unknown  disturbances, we add random noise $\tilde{w}(t)$ on $\tilde{q}(t)$ into the testing system, which takes the form of $\tilde{w}(t)=[20, 20, 20, 100]\times \text{rand}(t)\ \text{MW}$
%\begin{align}
%	\tilde{w}(t)=\left [ 
%	\begin{array}{c}
%		20\times \text{rand}(t)\\
%		20\times \text{rand}(t)\\ 
%		20\times \text{rand}(t)\\
%		100\times \text{rand}(t)
%	\end{array}  \right ] \text{MW} \nonumber
%\end{align}
, with rand(t) as a function generating a random number between $[0,\ 1]$ at time $t$. In the simulation, a random number is generated every 0.01s.
The load control command and the power variations are given Figure \ref{fig:noise_input_variation}. As the frequency of external disturbance is quite high, the internal model control is not able to follow it accurately. As a consequence, there exist obvious tracking errors. The system frequency is shown in  Figure \ref{fig:noise_frequency}. The inset zooms into the frequency dynamics between 140s-160s, when the system converges to the steady state. The maximal frequency deviation is smaller than 0.003Hz, demonstrating that the unknown disturbance is well attenuated by the proposed controller.

\begin{figure}[t]
	\centering
	\includegraphics[width=0.49\textwidth]{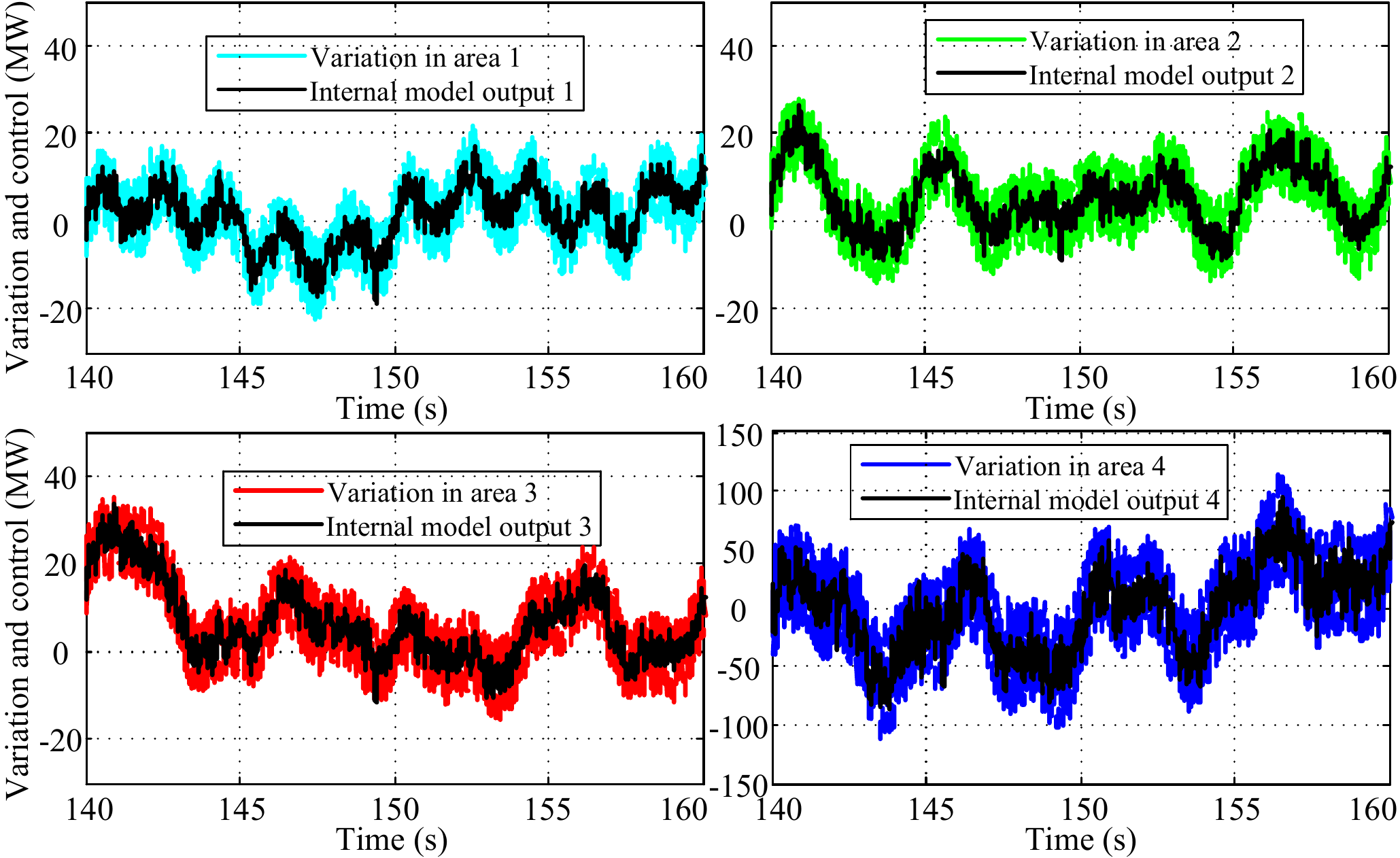}
	\caption{Variation and internal model output with noise}
	\label{fig:noise_input_variation}
\end{figure}

%\begin{figure}[t]
%	\centering
%	\includegraphics[width=0.49\textwidth]{noise_error}
%	\caption{Error between variation and internal model output}
%	\label{fig:noise_error}
%\end{figure}

\begin{figure}[t]
	\centering
	\includegraphics[width=0.32\textwidth]{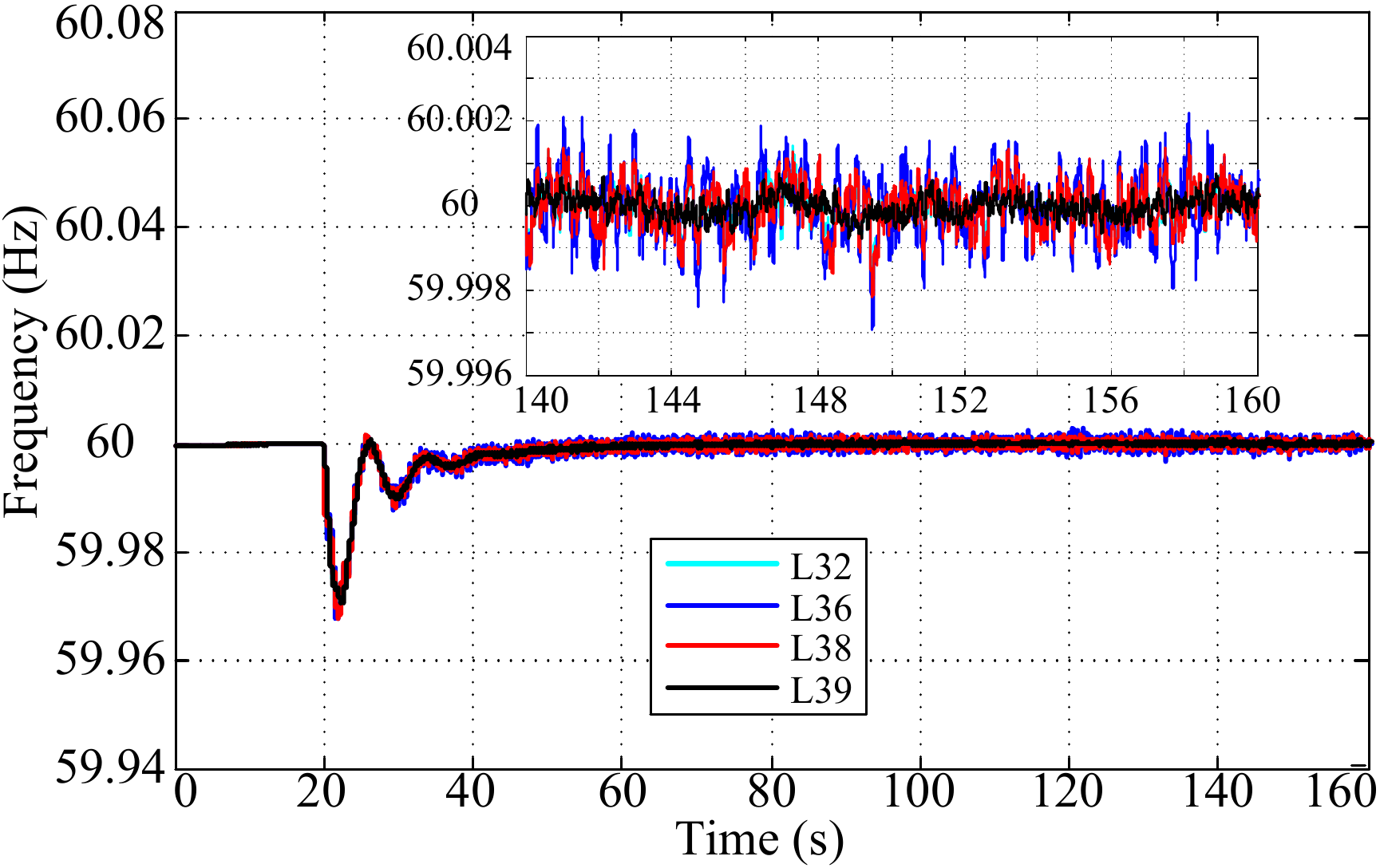}
	\caption{Frequency under external noise}
	\label{fig:noise_frequency}
\end{figure}

{\color{black}
\subsection{Simulation with Real Data}
\begin{figure}[t]
	\centering
	\includegraphics[width=0.48\textwidth]{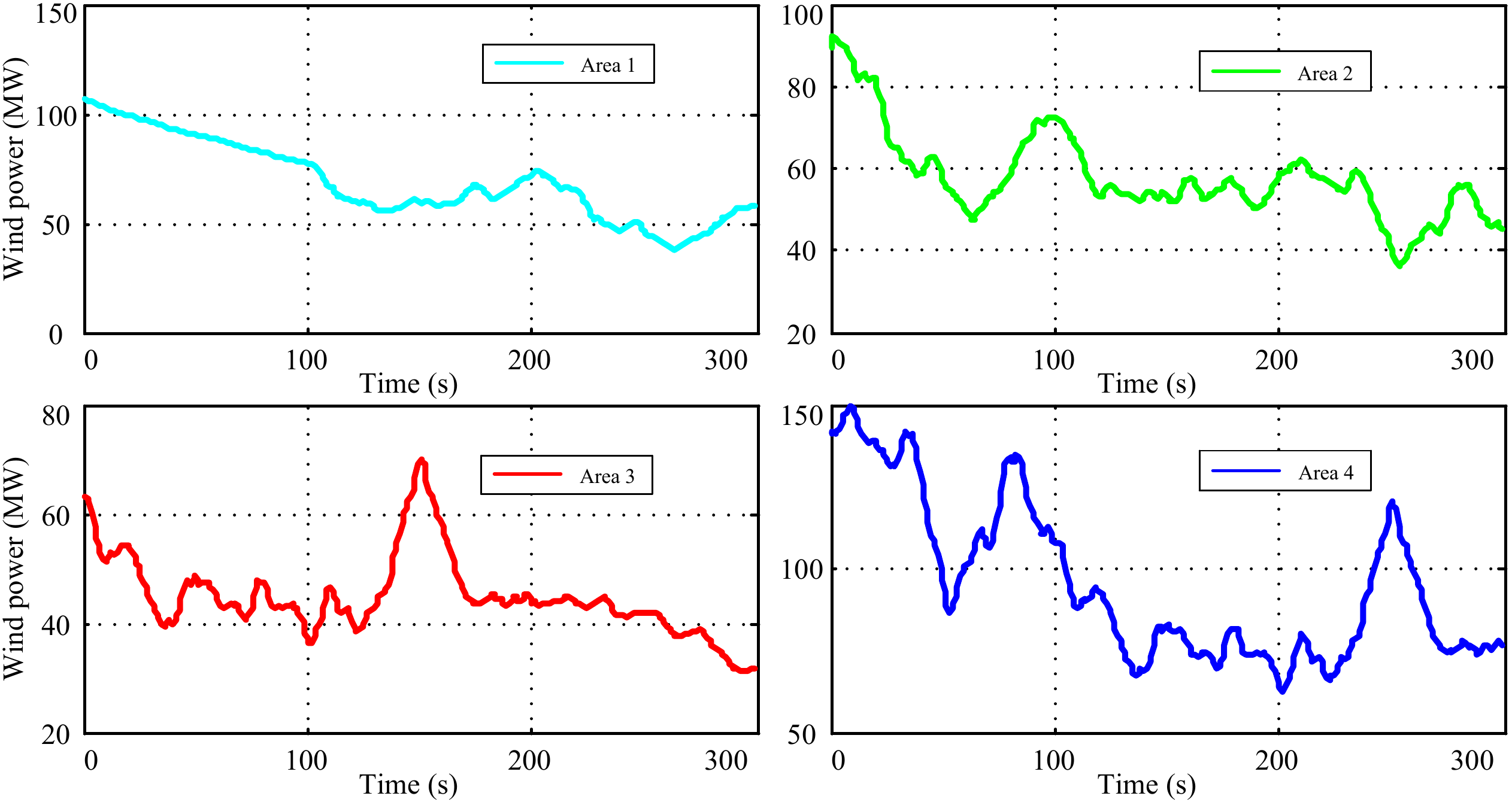}
	\caption{Wind power in each area}
	\label{fig_winddata}
\end{figure}

In this subsection, we use 300s data points for each area (one data point per second) to illustrate the effectiveness of our controller, which come from a real offshore wind farm. The data is available via the link \cite{website2018wind}. Due to agreement with the data provider, its is for personal use only. The wind power in each area is shown in Figure \ref{fig_winddata}, which is added in the simulation at $t=10$s. The power prediction in each area, i.e. $\overline P_j^{in}$, is $(72, 60, 49, 120)$MW respectively. The frequency dynamics using controller \eqref{controller_constant} and \eqref{controller_varying} with the real data are given in Figure \ref{fig_frequency_realdata}. Similar to that in Figure \ref{fig:load_frequency}, the frequency under the controller \eqref{controller_constant} varies and cannot be restored to the nominal value due to the variation of wind power. On the contrary, the frequency is very smooth when controller \eqref{controller_varying} is used. 
The performance of  controllable load tracking wind power variation in each area is given in Figure \ref{fig_error_realdata}. We can find that the controllable loads still coincide with the variations with high accuracy under the proposed controller. 

\begin{figure}[t]
	\centering
	\includegraphics[width=0.3\textwidth]{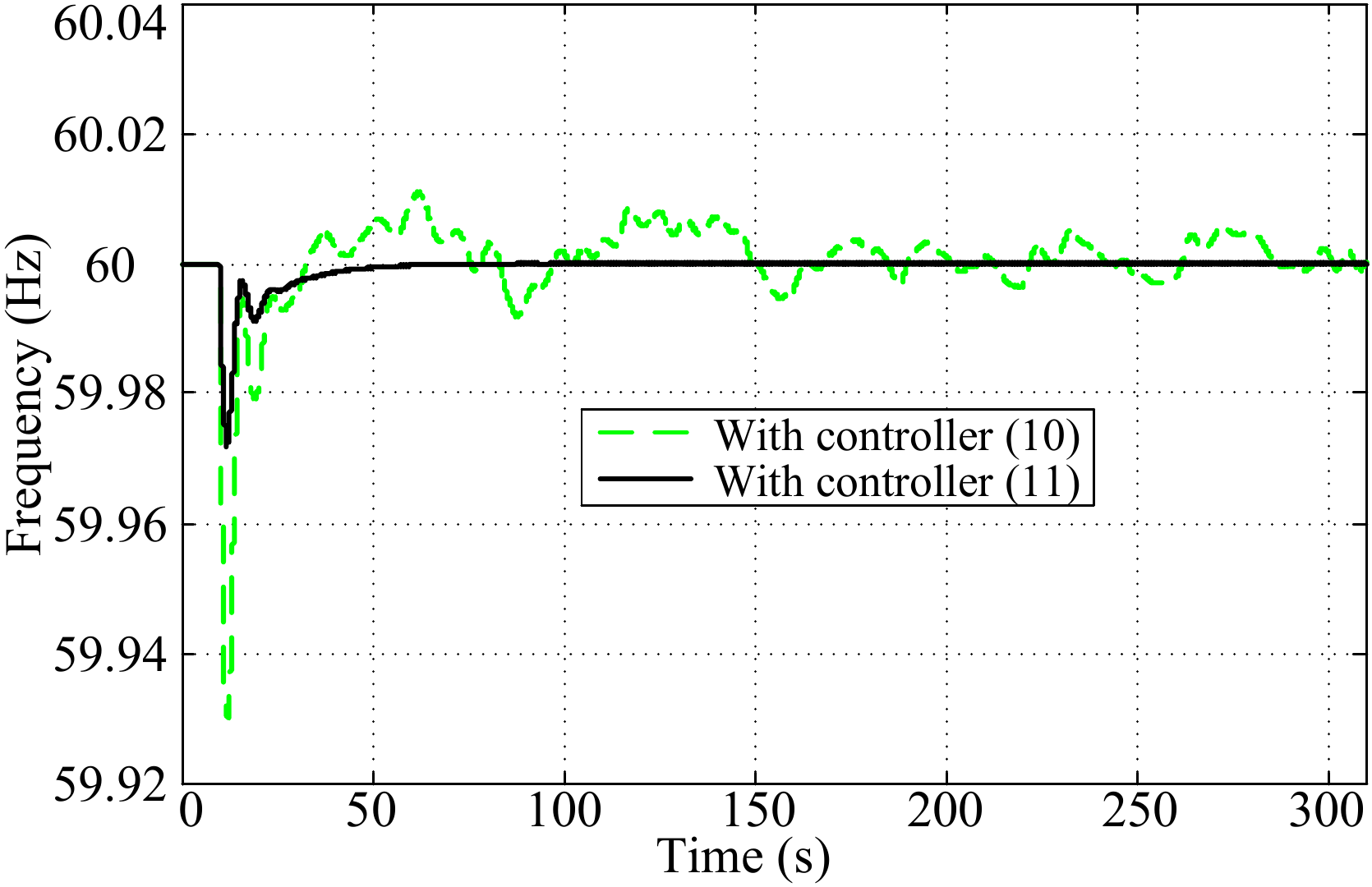}
	\caption{Frequency dynamics with real data}
	\label{fig_frequency_realdata}
\end{figure}

\begin{figure}[t]
	\centering
	\includegraphics[width=0.48\textwidth]{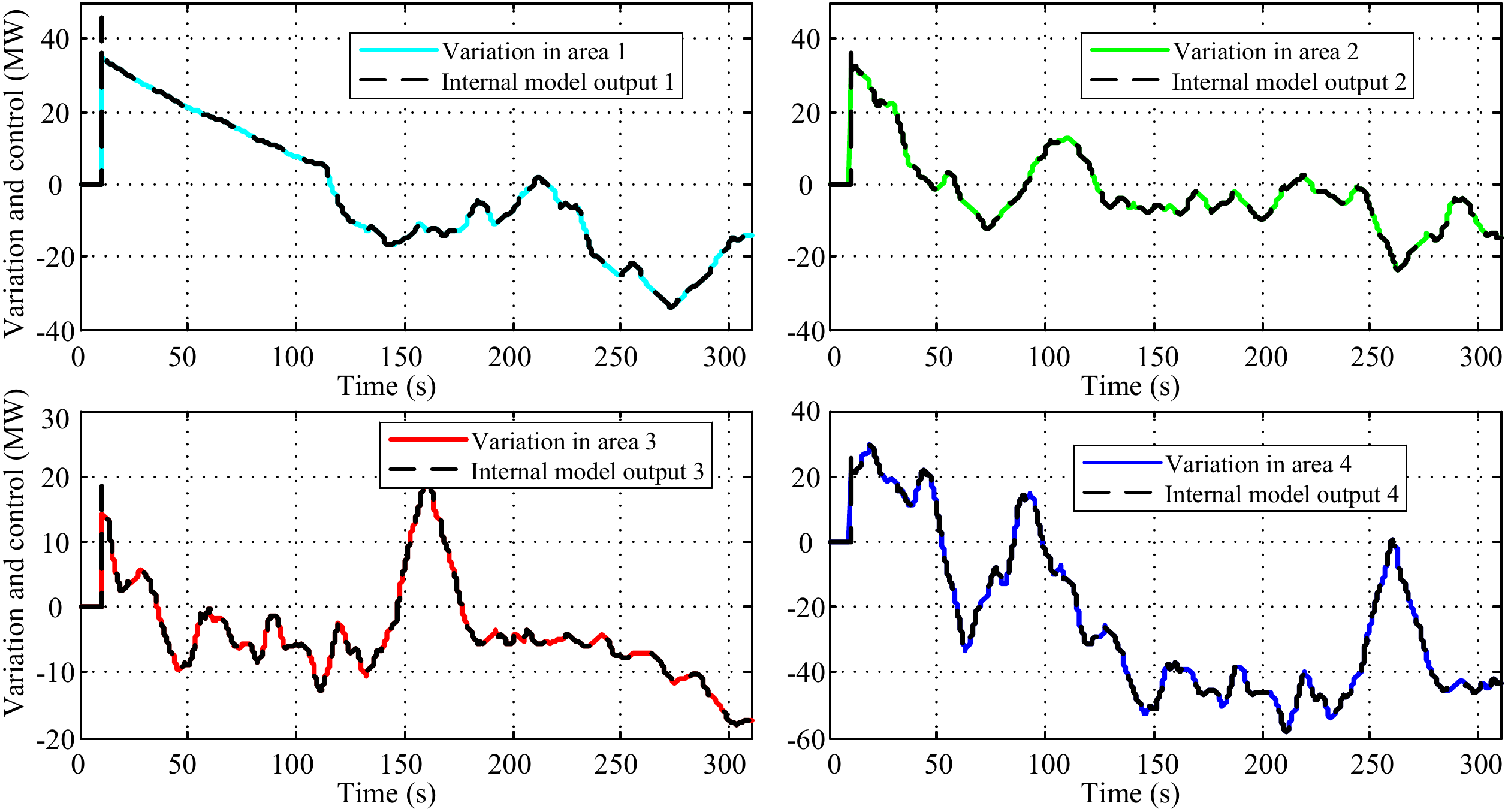}
	\caption{Variation and internal model output with real wind data}
	\label{fig_error_realdata}
\end{figure}

\subsection{Comparison with Existing Control Methods}

First, we compare the proposed method \eqref{controller_varying} with conventional PI control. In the PI control, the control command is $P_j^l=K_P\omega_j+K_I\int\omega_j\ dt$, where $K_P$ in each area are $(500, 500, 500, 800)$, and $K_I$ are $(2500, 2500, 2500, 4000)$. The frequency dynamics are given in the left part of Figure \ref{fig_PI}, where the inset is the enlarged version. It is shown that the frequency nadir is much larger than our method and the variation cannot be eliminated. This result demonstrates the superiority of the proposed method to the traditional PI control. 
\begin{figure}[t]
	\centering
	\includegraphics[width=0.49\textwidth]{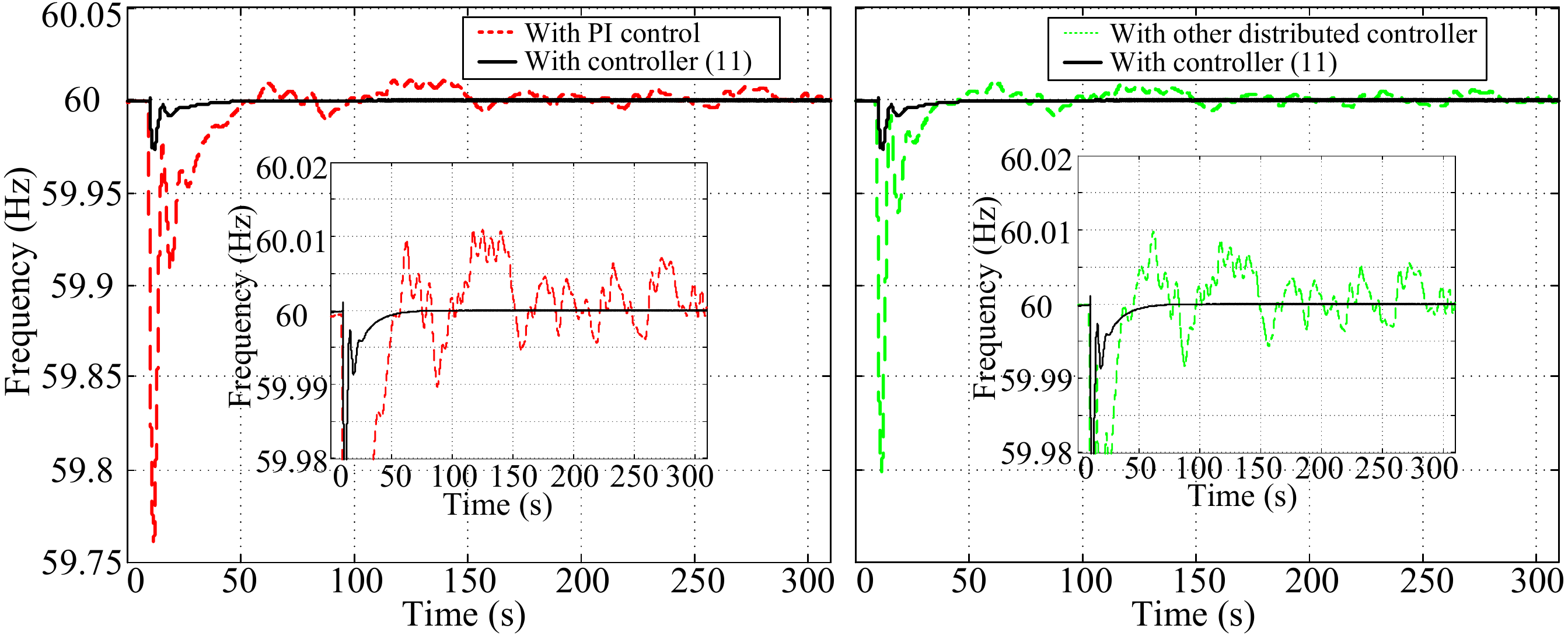}
	\caption{Frequency dynamics compared with PI control (left) and other distributed controller (right)}
	\label{fig_PI}
\end{figure}

We also compare the proposed method \eqref{controller_varying} with the distributed controller in \cite{Mallada-2017-OLC-TAC}. To make a valid comparison, we do not consider line constraints when using controller in \cite{Mallada-2017-OLC-TAC}, and the objective function is same with this paper. The frequency dynamics are given in the right part of Figure \ref{fig_PI}.
Similarly, the frequency variation is not eliminated, demonstrating the superiority of our controller in coping with unknown and time-varying power imbalance. }

\section{Conclusion}
{\color{black}
	In this paper, we have addressed the distributed frequency control problem of power systems in the presence of unknown and time-varying power imbalance.  We have decomposed power imbalance into three parts at different timescales: the known steady part, the unknown low-frequency variation and the unknown high-frequency residual. Then the distributed frequency control problem at the three different timescales are solved in a unified control framework composed of three timescales: 
	\begin{enumerate}
		\item \textbf{The slow timescale:} designing a consensus-based distributed control to allocate the steady part of power imbalance economically; 
		\item \textbf{The medium timescale:} devising an internal model control to accurately track and compensate for the time-varying unknown power imbalance locally; 
		\item \textbf{The fast timescale:} using the  $L_2$-gain inequality to show the robustness of the controller against uncertain disturbances and parameters.
	\end{enumerate}
	We have conducted numerical experiments using data of  the New England system and real-world wind farms. The empirical results show that our distributed controller can mitigate the frequency fluctuation caused by the integration of large uncertain and time-varying renewable generation. The test results also confirm that our controller outperforms existing ones.

This paper intends to provide a systematic approach to deal with unknown and time-varying power imbalance in an economic manner. Besides renewable generations, power oscillations and malicious attacks on controllers can also lead to unknown and time-varying power variation in power system operation. The proposed method could be extended to cope with such problems, which are among our future studies. 
}

\section*{References}

\bibliography{mybib,PowerRef-201202}

\end{document}